\documentclass[11pt]{article}

\usepackage{pgf}

\usepackage[english]{babel}
\usepackage[utf8]{inputenc}
\usepackage{amssymb}
\usepackage{graphicx}
\usepackage{lineno,hyperref}
\modulolinenumbers[5]

\usepackage{amsmath,amssymb}
\usepackage{geometry}
\usepackage{array}
\usepackage{xspace}

\usepackage{amsmath,amsfonts,amssymb,amsopn,amsthm}
\theoremstyle{plain}

\newtheorem{lemma}{Lemma}
\newtheorem{proposition}{Proposition}
\newtheorem{corollary}{Corollary}
\newtheorem{theorem}{Theorem}
\theoremstyle{definition}
\newtheorem{example}{Example}

\theoremstyle{remark}
\newtheorem{remark}{Remark}

\newcommand{\F}{{\mathbb F}}

 \newcommand{\sS}{{\mathcal S}}
 \newcommand{\sC}{{\mathcal C}}
 
 \newcommand{\sN}{{\mathcal N}}
 \newcommand{\sW}{{\mathcal W}}
\newcommand{\be}{\begin{eqnarray}}
\newcommand{\ee}{\end{eqnarray}}
\newcommand{\nn}{{\nonumber}}

\newcommand{\Tr}{{\rm Tr}}

\makeatother

\newcommand{\wt}{W_H}

\newcommand{\Wa}[1]{\sW_f}

\newcommand{\supp}{{\rm supp}}

\begin{document}

\title{Minimal Linear Codes From Weakly Regular Plateaued Balanced Functions
}


\author{Ahmet S{\i}nak       
}


\author{ Ahmet S{\i}nak{\thanks{ 
              Department of mathematics and computer science, Necmettin Erbakan University, 42090, Konya, Turkey and LAGA, Universities of Paris VIII and Paris XIII, CNRS, UMR 7539, Paris, France. 
Email: sinakahmet@gmail.com }}}


\maketitle

\begin{abstract}
Linear codes have diverse applications in secret sharing schemes, secure two-party computation,  association schemes,  strongly regular graphs, authentication codes and communication. There are a large number of linear codes with  few weights  in the literature, but a little  of them are minimal. In this paper, we are using for the first time weakly regular plateaued  balanced functions over the finite fields of odd characteristic in the second generic construction method of linear codes. The main results of this paper are stated below. 
We first construct several three-weight and four-weight  linear codes with flexible parameters  from weakly regular plateaued  balanced functions.
It is worth noting that  the (almost) optimal codes may be obtained from these functions.  
We next  observe that all  codes obtained in this paper  are minimal, thereby they can be directly employed to construct   secret sharing schemes with high democracy. Finally, the democratic  secret sharing schemes  are obtained from the dual codes of our minimal codes.
\end{abstract}

\section{Introduction}
Let  $\F_{p^n}$ represent  the finite field with $p^n$ elements, where  $p$ is  a prime number and  $n$ is a positive integer. 
The finite field   $\mathbb{F}_{p^n}$  can be seen as  an $n$-dimensional vector space over $\F_{p}$, and  denoted  by  $\F_p^n$. 
  An $\left[n,k,d\right]_{p}$  linear code $\sC$  over $\mathbb {F}_{p}$ is a $k$-dimensional   linear subspace of   $\F_p^n$ with  length $n$, dimension $k$ and minimum Hamming distance $d$. The Hamming weight of a codeword $\bold v=(v_0,\ldots, v_{n-1})\in \sC$, denoted by $\wt(\bold v)$,  is defined as  the size of  the set  $\supp(\bold  v)=\{0 \leq i\leq n-1: v_i\not=0 \}.$ 

 Let $A_\omega$ denote the number of codewords  in $\sC$  with Hamming weight $\omega$. Then the sequence $(1,A_1, \ldots, A_n)$ represents the \textit{ weight distribution}  and the polynomial $1+A_1y + \cdots + A_ny^n$ shows  the \textit{ weight enumerator} of the $n$-length code $\sC$. 
The code $\sC$ is said to be a  \textit{$t$-weight code} if  $\#\{1\leq \omega \leq n \colon\, A_\omega\neq 0\}=t$.
A \textit{generator matrix} $G$ of  $\sC$ is a $k\times n$ matrix whose rows form a basis for the code $\sC$. 
The \textit{dual code} of  $\sC$ is defined as
$\sC^\perp=\{\mathbf u \in\mathbb {F}_{p}^n  \colon\, \mathbf u \cdot \mathbf v=\mathbf 0 \mbox{ for all } \mathbf v\in\sC\}$ with length $n$ and dimension $n-k$, 
where $``\cdot"$ is the standard inner product on $\mathbb {F}_{p}^n$.

\par  
For the codewords $\bold u,\bold v\in \sC$,  if   $ \supp(\bold u) $ includes   $\supp(\bold v)$, then it is said that $\bold u$  covers  $\bold v$.
A nonzero codeword $\bold u$  of   $\sC$  is called \textit{minimal} codeword if $\bold u$  covers only the codeword $i\bold u$ for all $i\in\F_p$. 
Indeed, a linear  $\sC$ is called \textit{minimal linear code} if every nonzero codeword of $\sC$ is minimal codeword.  
Minimal linear codes have an interesting application in secret sharing scheme (SSS).
 In SSS, a set of participants who can reconstruct the secret value $s$ from their shares is said to be  \textit{an access set}. 
Besides,  an access set  is said to be  \textit{minimal access set} if  none of its proper subset  can reconstruct  $s$ from their shares.
The \textit{access structure} of a SSS is described as the  set of all access sets.
It is worth pointing out that we have a one-to-one match-up between the set of minimal codewords of the dual code $\sC^\perp$ and  the set of minimal access sets of SSS based on   $\sC $.  

Linear codes have diverse applications in    secret sharing schemes, secure two-party computation, association schemes,  strongly regular graphs, authentication codes, communication, data storage devices and consumer electronics. 
One of the well-known construction methods of linear codes is based on functions over finite fields.
This construction method is an interesting problem in coding theory. In the literature,  a larger number of linear codes with desirable parameters have been constructed from some special cryptographic functions such as  quadratic functions  \cite{tang2017linear},  (weakly regular) bent functions \cite{carlet1998codes,ding2015linear,ding2014binary,ding2015class,tang2016linear,DCCwu,zhou2016linear},  weakly regular plateaued functions \cite{mesnager2017linear,IEEE}, almost bent functions \cite{ding2016construction},  almost perfect nonlinear (APN) functions \cite{li2014weight,zeng2012triple}  and perfect nonlinear functions  \cite{carlet2005linear}. Very recently, weakly regular plateaued unbalanced functions have been  used in  \cite{IEEE}  to obtain  minimal linear codes with flexible  parameters.  Within this framework, we  benefit from weakly regular plateaued balanced functions   in order  to construct further minimal linear codes with different parameters over the finite fields of odd characteristic.  
 
The organization of the paper is  given as follows. Section \ref{WRP}  gives some  results on weakly regular plateaued balanced functions.  
In Section \ref{Constructions}, we obtain three-weight  and four-weight linear codes  from these functions   over  the finite field of odd characteristic. It is remarkable  that the  punctured three-weight optimal codes  are obtained in Examples \ref{Example0Punc} and \ref{Example00Punc}.
In Section \ref{SectionSSS}, we first observe the constructed codes are minimal,
 and  then define   the access structures of the SSS based on their dual codes.

\section{Weakly regular plateaued functions}\label{WRP}
This section introduces some useful results on exponential sums of  weakly regular plateaued balanced functions.

\subsection{Some results on weakly regular plateaued functions}
We first give a necessary background and some  results on weakly regular plateaued functions. For a set $S$,  $\# S$ represents the size of $S$ and $S^\star$ denotes $S\setminus \{0\}$.
The symbol $\eta_0$ represents  the \textit{quadratic character} of $\F_p^{\star}$,  and    $\eta_0(-1)p$ is denoted by $p^*$.
 The set of all  \textit{squares}  in  $\F_p^{\star}$ is denoted by  $SQ$ and   the set of all \textit{non-squares}  is denoted by $NSQ$.
Throughout this paper,  $f$ is a function from  $\mathbb {F}_{p^m}$  to $\mathbb {F}_{p}$ for an odd prime  $p$ and  a positive integer $m$.
The  trace of $\alpha\in\mathbb {F}_{p^m}$ over $\mathbb {F}_{p}$ is defined by
$\Tr^m(\alpha)=\alpha+\alpha^{p}+\alpha^{p^{2}}+\cdots+\alpha^{p^{m-1}}$.
The  \textit{Walsh transform} of $f$ is defined by
\be\nn
\Wa {f}(\omega )=\sum_{x\in  \mathbb {F}_{p^m}} {\zeta_p}^{{f(x)}-\Tr^m (\omega  x)},
\ee 
where $\zeta_p$ is a  primitive $p$-th root of unity.
A function  $f$ is said to be \textit{balanced} over $\F_p$ if   $\Wa f(0)=0$; otherwise, $f$ is  \textit{unbalanced}. 
\par The notion of  plateaued Boolean functions was first introduced   by Zheng and Zhang  \cite{zheng1999plateaued}.  In characteristic $p$,
a function $f$ is  called  $p$-ary \emph{$s$-plateaued}   if $|\Wa {f}(\omega)|^2\in\{0,p^{m+s}\}$ for every $\omega\in \F_{p^m}$, with $0\leq s\leq m$. In particular,  a $0$-plateaued function is  the \emph{bent function}.  
The \textit{Walsh support} of a plateaued   $f$ is defined as the set   $\sS_f =\{\omega\in  \F_{p^m} \colon\,  |\Wa {f}(\omega)|^2= p^{m+s}\}$. 
 The absolute Walsh distribution of   a plateaued function can be derived from the \textit{Parseval identity}.  
 \begin{lemma}\label{SupportLemma}
If $f$ is an $s$-plateaued function over $\F_{p^m}$, then for $\omega \in\F_{p^m}$, $| \Wa {f}(\omega)|^2$ takes $p^{m-s}$ times the value $p^{m+s}$ and $p^m-p^{m-s}$ times the value $0$.
\end{lemma}

Recently, motivated by \cite[Theorem 2]{hyun2016explicit}, Mesnager et al. \cite{mesnager2017WCC,mesnager2019linear} have  defined the subclass of plateaued functions.
An $s$-plateaued $f$ is called  \emph{weakly regular}  if  
\be\label{PlateauedWalshh}
\Wa {f}(\omega)\in \left\{ 0, up^{\frac{m+s}2}\zeta_p^{f^{\star}(\omega)}\right \}, \; \forall \omega\in \F_{p^m},
\ee
where $u\in\{\pm 1,\pm i \}$ and  $f^{\star}$ is  a $p$-ary function  over $\F_{p^m}$ with $f^{\star}(\omega)=0$ for all $\omega\in  \F_{p^m} \setminus  \sS_f$. We remark that $f$ is said to be a \emph{non-weakly regular plateaued function} when  $u$ in (\ref{PlateauedWalshh}) depends on $\omega$. Notice that  $f^{\star}(0)=0$ if  $f$ is a  plateaued balanced function.

\begin{lemma}\label{WalshFact} \cite{mesnager2019linear}
If  $f$ is a weakly regular  $s$-plateaued function over $\F_{p^m}$, then
for  every $\omega\in \sS_f$,
$ \Wa {f}(\omega)=\epsilon \sqrt{p^*}^{m+s} \zeta_p^{f^{\star}(\omega)},$
 where $\epsilon\in\{\pm 1\}$ is    the sign of  $ \Wa f$ and  $f^{\star}$ is a $p$-ary function over  $\sS_f$. 
\end{lemma}
  
Very recently,  Mesnager et al.  \cite{IEEE} have   denoted by  $\textit{WRP}$  the set of  weakly regular plateaued  unbalanced functions with the following conditions.
Within the same framework,  we now assume that  $f:\F_{p^m}\to\F_{p}$ is a weakly regular   $s$-plateaued balanced  function, with $0\leq s\leq m$, and we denote by  $\textit{WRPB}$  the class of such functions satisfying the following   two homogeneous conditions:
 \begin{itemize}
\item $f(0)=0,$ 
\item   $f(ax)=a^tf(x)$  for every $a\in\F_p^{\star}$ and $x\in \F_{p^m}$,   where $t$ is  a positive even integer with $\gcd(t-1,p-1)=1$.
\end{itemize}
\begin{remark}
This paper uses for the first time the  plateaued functions from the class $\textit{WRPB}$ to construct  new minimal linear codes with flexible parameters. 
\end{remark}
We need in the subsequent section the following results that can be derived   from  \cite[Lemma 6 and Proposition 2]{IEEE}.  

 \begin{lemma} \label{Walshsupport}
Let $\omega\in\F_{p^m}$ and  $f\in \textit{WRPB}$ with  $ \Wa {f}(\omega)=\epsilon \sqrt{p^*}^{m+s} \zeta_p^{f^{\star}(\omega)}$. Then for every $z\in\F_p^{\star}$,  we have
$z\omega\in\sS_f$ when  $\omega\in\sS_f$, and  
 $z\omega\in\F_{p^m}\setminus\sS_f$; otherwise.
\end{lemma}
 
\begin{proposition}
Let $f\in \textit{WRPB}$  with  $ \Wa {f}(\omega)=\epsilon \sqrt{p^*}^{m+s} \zeta_p^{f^{\star}(\omega)}$ for every  $\omega\in \sS_f$.  Then,  we have $f^{\star}(a \omega)=a^lf^{\star}(\omega)$  for every $a\in\F_p^{\star}$ and $\omega \in \sS_f$, where $l$ is a  positive even integer  with $\gcd(l-1,p-1)=1$.
\end{proposition}

We end this subsection with giving a brief introduction to the quadratic functions  (see for example \cite{helleseth2006monomial}).
Recall that  every quadratic function from  $\F_{p^m}$ to $\F_p$ having no linear term can be represented by
\be\label{Quadratic}
Q(x)=\sum_{i=0}^{\lfloor m/2\rfloor} \Tr^m(a_ix^{p^i+1}),
\ee
where  $a_i\in\F_{p^m}$ for $0\leq i\leq \lfloor m/2\rfloor$ and $\lfloor x \rfloor$ represents the largest integer less than or equal to $x$.
Let  $A$ be  a corresponding $m\times m$ symmetric matrix with $Q(x)=x^TAx$ as in   \cite{helleseth2006monomial} and  $L$ be a corresponding linearized polynomial over   $\F_{p^m}$  defined as
$
L(z)=\sum_{i=0}^{l}(a_iz^{p^i} +a_i^{p^{m-i}} z^{p^{m-i}}).
$
 The set of linear structures of quadratic function $Q$ is the kernel of $L$, defined as 
\be\label{LinearizedKernel}
\ker_{\F_p}(L)=\{z\in \F_{p^m}  \colon\, Q(z+y)=Q(z)+Q(y), \forall y\in \F_{p^m}\},
\ee
which is an $\F_p$-linear subspace of $ \F_{p^m}$.
 Let  the dimension of $\ker_{\F_p}(L)$  be  $s$ with $0\leq s\leq m$. Notice that by \cite[Proposition 2.1]{hou2004solution}, the rank of $A$ equals $m-s$. 
It was shown in  \cite{helleseth2006monomial} that a quadratic function $Q$ is bent if and only if $s=0$; equivalently,  $A$ is nonsingular, that is, $A$  has full rank. Hence 
we have the following natural consequence (see \cite[Proposition 2]{helleseth2006monomial} and \cite[Example 1]{mesnager2015results}).

\begin{proposition}  
Any quadratic function $Q$ is an $s$-plateaued if and only if the dimension of the kernel of $L$ defined as  in (\ref{LinearizedKernel}) equals $s$; equivalently, the rank of $A$ equals $m-s$.  
\end{proposition}

One can derive from \cite[Proposition 1]{helleseth2006monomial} and \cite[Theorem 4.3]{ccesmeliouglu2012construction}   the following reasonable fact. 
\begin{proposition}
 The sign of the Walsh transform of quadratic functions does not  depend on inputs which means that every quadratic function is a weakly regular plateaued function. Namely, there is no quadratic  non-weakly regular plateaued  function.  
\end{proposition}

\begin{remark}
All quadratic balanced functions  are included in the class \textit{WRPB}.  
\end{remark}

 \subsection{Exponential sums from weakly regular plateaued functions}
In this subsection, we give some results on exponential sums about weakly regular plateaued balanced functions.

\begin{lemma}\cite{IEEE}\label{Lemmag}  
Let $f:\F_{p^m}\to\F_{p}$ be a weakly regular $s$-plateaued function with  $ \Wa {f}(\omega)=\epsilon \sqrt{p^*}^{m+s} \zeta_p^{f^{\star}(\omega)}$ for   $\omega\in \sS_f$,  where   $f^{\star}$ is a $p$-ary function over  $\sS_f$.
For $a\in\F_p$, define 
$\sN_{f^{\star}}(a)=\#\{\omega\in\sS_f  \colon\, f^{\star}(\omega)=a\}.$ Then we have
\be\nn
\sN_{f^{\star}}(a)=\left\{\begin{array}{ll}
p^{m-s-1} + \epsilon \eta_0^{m+1}(-1)(p-1) \sqrt{p^*}^{m-s-2},  & \mbox{ if } a=0, \\
p^{m-s-1} -\epsilon \eta_0^{m+1}(-1) \sqrt{p^*}^{m-s-2},& \mbox{ if } a\in\F_p^{\star}
 \end{array}\right.
\ee
 when $m-s$ is even; otherwise,
\be\nn
\sN_{f^{\star}}(a)=\left\{\begin{array}{ll}
p^{m-s-1},  & \mbox{ if } a=0, \\
p^{m-s-1} +\epsilon  \eta_0^{m}(-1) \sqrt{p^*}^{m-s-1},& \mbox{ if } a \in SQ,\\
p^{m-s-1} - \epsilon\eta_0^m(-1)   \sqrt{p^*}^{m-s-1}, & \mbox{ if }  a\in NSQ.
 \end{array}\right.
\ee
\end{lemma}

\begin{lemma}\label{LemmaA}
Let $f\in \textit{WRPB}$. For   $\omega\in\F_{p^m}^{\star}$, define
$$
A=\sum_{y,z\in\F_p^{\star}} \sum_{x\in\F_{p^m}}\zeta_p^{yf(x)-z\Tr^m(\omega x)}.
$$
 Then for every $\omega\in\F_{p^m}^{\star} \setminus \sS_f$ we have $A=0$, and for every  $\omega\in \sS_f$  
\be\nn
A=\left\{\begin{array}{ll}
  \epsilon  (p-1)^2 \sqrt{p^*}^{m+s},  & \mbox{ if } f^{\star}(\omega)=0, \\
- \epsilon  (p-1) \sqrt{p^*}^{m+s},& \mbox{ if } f^{\star}(\omega)\neq 0
 \end{array}\right.
\ee
 when $m+s$ is even; otherwise,  
\be\nn
A=\left\{\begin{array}{ll}
0,  & \mbox{ if } f^{\star}(\omega)=0, \\
 \epsilon  \eta_0(f^{\star}(\omega)) (p-1) \sqrt{p^*}^{m+s+1},& \mbox{ if } f^{\star}(\omega)\neq 0.
 \end{array}\right.
\ee
\end{lemma}
\begin{proof}
The proof can proceed  by using the same  arguments of the proof of \cite[Lemma 12]{IEEE}.
\end{proof}

\begin{lemma}\label{Lemma0}
Let $f\in \textit{WRPB}$.  For $\omega\in\F_{p^m}^{\star}$, define 
$\sN_{0}(\omega)=\#\{x\in\F_{p^m}  \colon\, f(x)=0 \mbox{ and } \Tr^m(\omega x)=0\}.$
Then for every $\omega\in\F_{p^m}^{\star} \setminus \sS_f$  we have
$\sN_{0}(\omega)= p^{m-2},$
 and for every  $\omega\in \sS_f,$  
\be\nn
\sN_{0}(\omega)=\left\{\begin{array}{ll}
p^{m-2} + \epsilon  (p-1)^2\sqrt{p^*}^{m+s-4},  & \mbox{ if } f^{\star}(\omega)=0, \\
p^{m-2}- \epsilon  (p-1)\sqrt{p^*}^{m+s-4},& \mbox{ if } f^{\star}(\omega)\neq 0
 \end{array}\right.
\ee
when $m+s$ is even; otherwise,
\be\nn
\sN_{0}(\omega)=\left\{\begin{array}{ll}
p^{m-2},  & \mbox{ if } f^{\star}(\omega)=0, \\
p^{m-2}+\epsilon   (p-1)  \sqrt{p^*}^{m+s-3},& \mbox{ if } f^{\star}(\omega)\in SQ,\\
p^{m-2}-\epsilon   (p-1)  \sqrt{p^*}^{m+s-3},& \mbox{ if } f^{\star}(\omega)\in NSQ.
 \end{array}\right.
\ee
\end{lemma}
\begin{proof} By the definition of $\sN_{0}(\omega)$ and the fact that $f$ is balanced, we have
\be\nn \label{BBB}
\begin{array}{ll}
\sN_{0}(\omega)=p^{m-2}
+p^{-2}\displaystyle \sum_{y,z\in\F_p^{\star}}\sum_{x\in\F_{p^m}}  \zeta_p^{yf(x)-z\Tr^m(\omega x)}.
\end{array}
\ee
Then, the proof is ended from Lemma \ref{LemmaA}.
\end{proof}

 \begin{lemma} \label{LemmaAA}
Let $f\in \textit{WRPB}$. For $\omega\in\F_{p^m}^{\star}$, define
$$
A=\sum_{y,z\in\F_p^{\star}} \sum_{x\in\F_{p^m}}\zeta_p^{y^2f(x)-z\Tr^m(\omega x)}.
$$
Then for every $\omega\in\F_{p^m}^{\star} \setminus \sS_f$ we have $A=0$, and for every $\omega\in \sS_f$
\be\nn
A=\left\{\begin{array}{ll}
  \epsilon  (p-1)^2 \sqrt{p^*}^{m+s},  & \mbox{ if } f^{\star}(\omega)=0, \\
 \epsilon  (p-1) \sqrt{p^*}^{m+s}(\sqrt{p^*}-1),& \mbox{ if } f^{\star}(\omega)\in SQ,\\
- \epsilon  (p-1) \sqrt{p^*}^{m+s}(\sqrt{p^*}+1),& \mbox{ if }  f^{\star}(\omega)\in NSQ.
 \end{array}\right.
\ee
\end{lemma}
\begin{proof}
The proof can proceed by using the  arguments used in the proof of \cite[Lemma 15]{IEEE}.
\end{proof}

\begin{lemma}\label{LemmaSQ}
Let $f\in \textit{WRPB}$. For $\omega\in\F_{p^m}^{\star}$, define
\be\nn
\begin{array}{ll}
\sN_{sq}(\omega)&=\#\{x\in\F_{p^m}  \colon\, f(x)\in SQ \mbox{ and } \Tr^m(\omega x)=0\},\\
\sN_{nsq}(\omega)&=\#\{x\in\F_{p^m}  \colon\, f(x)\in NSQ \mbox{ and } \Tr^m(\omega x)=0\}.\\
\end{array}
\ee
Then for every $\omega\in\F_{p^m}^{\star} \setminus \sS_f$ we have  
$\sN_{sq}(\omega)=\sN_{nsq}(\omega)= \frac{1}{2}(p-1)p^{m-2}$. For every  $\omega\in \sS_f$  
\be\nn\begin{array}{ll} \vspace{.1 cm}
\sN_{sq}(\omega)&=\left\{\begin{array}{ll}
\frac{1}{2}(p-1)(p^{m-2} -\epsilon (p-1)\sqrt{p^*}^{m+s-4}),  & \mbox{ if } f^{\star}(\omega)=0 \mbox{ or } f^{\star}(\omega) \in NSQ, \\
\frac{1}{2}(p-1)(p^{m-2} + \epsilon(p+1)\sqrt{p^*}^{m+s-4}),& \mbox{ if } f^{\star}(\omega) \in SQ,
 \end{array}\right.\\

\sN_{nsq}(\omega)&=\left\{\begin{array}{ll}
\frac{1}{2}(p-1)(p^{m-2} -\epsilon (p-1)\sqrt{p^*}^{m+s-4}),  & \mbox{ if } f^{\star}(\omega)=0  \mbox{ or }  f^{\star}(\omega) \in SQ, \\
\frac{1}{2}(p-1)(p^{m-2} + \epsilon (p+1) \sqrt{p^*}^{m+s-4}),& \mbox{ if } f^{\star}(\omega) \in NSQ
 \end{array}\right.\\
\end{array}
\ee
 when $m+s$ is even; otherwise,
\be\nn\begin{array}{ll}\vspace{.1 cm}
\sN_{sq}(\omega)&=\left\{\begin{array}{ll}
\frac{1}{2}(p-1)(p^{m-2}+\epsilon\eta_0(-1)(p-1) \sqrt{p^*}^{m+s-3}),  & \mbox{ if } f^{\star}(\omega)=0, \\
\frac{1}{2}(p-1)(p^{m-2} - \epsilon  \sqrt{p^*}^{m+s-3}(\eta_0(-1)+1)),& \mbox{ if } f^{\star}(\omega) \in SQ,\\
\frac{1}{2}(p-1)(p^{m-2} -\epsilon  \sqrt{p^*}^{m+s-3}(\eta_0(-1)-1)),  & \mbox{ if }  f^{\star}(\omega) \in NSQ, \\
 \end{array}\right.\\

\sN_{nsq}(\omega)&=\left\{\begin{array}{ll}
\frac{1}{2}(p-1)(p^{m-2}-\epsilon\eta_0(-1)(p-1) \sqrt{p^*}^{m+s-3}),  & \mbox{ if } f^{\star}(\omega)=0, \\
\frac{1}{2}(p-1)(p^{m-2} +\epsilon  \sqrt{p^*}^{m+s-3}(\eta_0(-1)-1)),  & \mbox{ if } f^{\star}(\omega) \in SQ, \\
\frac{1}{2}(p-1)(p^{m-2} + \epsilon  \sqrt{p^*}^{m+s-3}(\eta_0(-1)+1)),  
& \mbox{ if } f^{\star}(\omega) \in NSQ.
 \end{array}\right.
 \end{array}
\ee
\end{lemma} 
\begin{proof}
From the proof of \cite[Lemma 14]{tang2016linear}, recalling that $f$ is balanced,
we have 
\be\nn\label{WWW}
\begin{array}{ll}
p^2 \sN_{0}(\omega)+p\sqrt{p^*}(\sN_{sq}(\omega)-\sN_{nsq}(\omega))=
p^{n} + \displaystyle \sum_{y,z\in\F_p^{\star}}\sum_{x\in\F_{p^m}}  \zeta_p^{y^2f(x)-z\Tr^n(\omega x)},
\end{array}
\ee
where $\sN_{0}(\omega)$ is given in Lemma \ref{Lemma0}.  We clearly have $\sN_{0}(\omega)+\sN_{sq}(\omega)+\sN_{nsq}(\omega)=p^{n-1}$.
Hence, combining these results,  the proof is  immediately completed  from  Lemmas \ref{Lemma0} and  \ref{LemmaAA}. 
\end{proof}

The following lemma   is a direct consequence of Lemma \ref{LemmaSQ}.

\begin{lemma}\label{Lemma12}
Let $f\in \textit{WRPB}$. For $\omega\in\F_{p^m}^{\star}$, define
$\sN_{1}(\omega)=\#\{x\in\F_{p^m}  \colon\, f(x)=1 \mbox{ and } \Tr^m( \omega x)=0\}$ and
$\sN_{2}(\omega)=\#\{x\in\F_{p^m}  \colon\, f(x)=2 \mbox{ and } \Tr^m( \omega x)=0\}.$
Then, 
\be\nn
\sN_{1}(\omega)=\frac{2\sN_{sq}(\omega)}{(p-1)} \mbox{ and } \sN_{2}(\omega)=\frac{2\sN_{nsq}(\omega)}{(p-1)}.
\ee
\end{lemma} 
 The following lemma can be deduced from the combination of Lemmas  \ref{Lemma0} and \ref{LemmaSQ}.
\begin{lemma}\label{LemmaSQ0}
Let $f\in \textit{WRPB}$. For $ \omega\in\F_{p^m}^{\star}$, define
\be\nn
\begin{array}{ll}
\sN_{(sq,0)} (\omega)&=\#\{x\in\F_{p^m}  \colon\, f(x)\in SQ\cup \{0\} \mbox{ and } \Tr^m( \omega x)=0\},\\
\sN_{(nsq,0)} (\omega)&=\#\{x\in\F_{p^m}  \colon\, f(x)\in NSQ\cup \{0\} \mbox{ and } \Tr^m( \omega x)=0\}.
\end{array}
\ee
Then for every $ \omega\in\F_{p^m}^{\star} \setminus \sS_f,$ we have  $\sN_{(sq,0)} (\omega)=\sN_{(nsq,0)} (\omega)=\frac{1}{2}(p+1)p^{m-2}$.
For every  $ \omega\in \sS_f,$  
\be\nn\begin{array}{ll}
\sN_{(sq,0)} (\omega)&=\left\{\begin{array}{ll}
\frac{1}{2}(p+1)p^{m-2} + \epsilon \frac{1}{2}(p-1)^2 \sqrt{p^*}^{m+s-4},& \mbox{ if } f^{\star}( \omega)=0  \mbox{ or }  f^{\star}( \omega) \in SQ,\\
\frac{1}{2}(p+1)(p^{m-2} -\epsilon (p-1)\sqrt{p^*}^{m+s-4}),& \mbox{ if } f^{\star}( \omega) \in NSQ,
 \end{array}\right.\\
\vspace{.1 cm}
\sN_{(nsq,0)} (\omega)&=\left\{\begin{array}{ll}
\frac{1}{2}(p+1)p^{m-2} + \epsilon \frac{1}{2}(p-1)^2 \sqrt{p^*}^{m+s-4},& \mbox{ if } f^{\star}( \omega)=0  \mbox{ or }  f^{\star}( \omega) \in NSQ,\\
\frac{1}{2}(p+1)(p^{m-2} -\epsilon (p-1)\sqrt{p^*}^{m+s-4}),& \mbox{ if } f^{\star}( \omega) \in SQ,
 \end{array}\right.\\
\end{array}
\ee
 when $m+s$ is even; otherwise,
\be\nn\begin{array}{ll}
\sN_{(sq,0)} (\omega)&=\left\{\begin{array}{ll}
\frac{1}{2}(p+1)p^{m-2}+\epsilon\eta_0(-1)\frac{1}{2}(p-1)^2 \sqrt{p^*}^{m+s-3},  & \mbox{ if } f^{\star}( \omega)=0, \\
\frac{1}{2}(p+1)p^{m-2} -  \epsilon  \frac{1}{2}(p-1)\sqrt{p^*}^{m+s-3}(\eta_0(-1)-1),& \mbox{ if } f^{\star}( \omega) \in SQ,\\
\frac{1}{2}(p+1)p^{m-2}-\epsilon  \frac{1}{2}(p-1) \sqrt{p^*}^{m+s-3}(\eta_0(-1)+1),  & \mbox{ if }  f^{\star}( \omega) \in NSQ, \\
 \end{array}\right.\\
\vspace{.1 cm}
\sN_{(nsq,0)} (\omega)&=\left\{\begin{array}{ll}
\frac{1}{2}(p+1)p^{m-2}-\epsilon\eta_0(-1)\frac{1}{2}(p-1)^2 \sqrt{p^*}^{m+s-3},  & \mbox{ if } f^{\star}( \omega)=0, \\
\frac{1}{2}(p+1)p^{m-2} +   \epsilon  \frac{1}{2}(p-1)\sqrt{p^*}^{m+s-3}(\eta_0(-1)+1),& \mbox{ if } f^{\star}( \omega) \in SQ,\\
\frac{1}{2}(p+1)p^{m-2}+\epsilon  \frac{1}{2}(p-1) \sqrt{p^*}^{m+s-3}(\eta_0(-1)-1),  & \mbox{ if }  f^{\star}( \omega) \in NSQ. \\
 \end{array}\right.
\end{array}
\ee
\end{lemma} 
\section{Linear codes from $f\in\textit{WRPB}$}\label{Constructions}
 
This section presents the flexible parameters of linear codes constructed from  weakly regular plateaued balanced functions.
  In the literature, there are several construction methods of linear codes  based on functions over finite fields.  As stated  by Ding  in  \cite{ding2016construction}, we can distinguish two  of them from the others, which are called the \textit{first} and \textit{second} generic construction methods based on functions. The first generic construction  is defined over $\mathbb{F}_p$ by 
\be\nn
\mathcal{C}(F)=\{{\bold c_{(a,b)}}=(\Tr^m(aF(x)+bx))_{x\in\mathbb{F}^*_{p^m}}  \colon\,  a,b\in\F_{p^m}\},
\ee
 where  $F$  is a polynomial over $\F_{p^m}$.
The code $\mathcal{C}(F)$  is a $p$-ary  linear code of length $p^m-1$ and  dimension at most $2m$. 
The second generic construction of linear codes from functions is defined  by assigning a subset $D=\{d_1,  \ldots, d_n\}$ of $\F_{p^m}$.
A $p$-ary linear code involving $D$ is defined by
\be\label{LinearCodes}
\sC_{D}=\{\bold   c_\omega= (\Tr{^m}(\omega d_1), \ldots, \Tr{^m}(\omega d_n))  \colon\,  \omega \in\mathbb{F}_{p^m}\},
\ee
whose length equals $n$ and  dimension at most $m$.   The subset $D$ is usually called the {\em defining set} of $\mathcal {C}_D$. 
The quality of the parameters depends on the selection of  the defining set $D$.
This method  was first proposed   by  Ding et. al \cite{ding2009class,ding2015linear,ding2007cyclotomic,ding2014binary,ding2015class}, and  a large number of (minimal) linear codes with perfect parameters have been obtained in these papers.
Furthermore, this   method has been widely studied in the literature, and 
several  (minimal) linear codes with few weights  have been constructed from cryptographic functions over finite fields (for example \cite{ding2016construction,mesnager2019linear,IEEE,tang2016linear,tang2017linear,zhou2016linear}). 
  We in this paper study the linear codes of the form (\ref{LinearCodes}) by selecting  the following defining sets 
\be\label{DefiningSets}
\begin{array}{ll}
D_{0}&=\{ x\in\F_{p^m}^{\star} \colon\, f(x)=0\},\\
D_{1}&=\{ x\in\F_{p^m} \colon\, f(x)=1\}, \\
D_{2}&=\{ x\in\F_{p^m} \colon\, f(x)=2\},\\
D_{(0,1)}&=\{ x\in\F_{p^m}^{\star} \colon\, f(x)\in\{0,1\}\}, \\
D_{(0,2)}&=\{ x\in\F_{p^m}^{\star} \colon\, f(x)\in\{0,2\}\}, \\
D_{(1,2)}&=\{ x\in\F_{p^m} \colon\, f(x)\in\{1,2\}\},\\
D_{sq}&=\{ x\in\F_{p^m} \colon\, f(x)\in SQ\},\\
D_{nsq}&=\{ x\in\F_{p^m} \colon\, f(x)\in NSQ\},\\
D_{(sq,0)}&=\{ x\in\F_{p^m}^{\star} \colon\, f(x)\in SQ\cup \{0\}\},\\
D_{(nsq,0)}&=\{ x\in\F_{p^m}^{\star} \colon\, f(x)\in NSQ\cup \{0\}\},
\end{array}
\ee
where $f\in \textit{WRPB}$. 
Since $f$ is a  balanced function with $f(0)=0$,  we  have
\be\label{length}
\begin{array}{ll}
\# D_{0}=p^{m-1}-1,\\ 
 \# D_{1}=\#  D_{2}=p^{m-1},\\
\# D_{(0,1)}= \#D_{(0,2)}=2p^{m-1}-1,\\
\# D_{(1,2)}=2p^{m-1},\\
\# D_{sq}=\#  D_{nsq}=p^{m-1} (p-1)/2,\\
\#  D_{(sq,0)}=\#  D_{(nsq,0)}=p^{m-1}(p+1)/2-1,
\end{array}
\ee
which are the lengths of the codes involving these sets. These different selections of the defining sets provide  new parameters for the  linear codes of the form  (\ref{LinearCodes}), which discover   several new classes of minimal linear codes with few weights. 
We first consider the defining set  $D_{0}$ of the form  (\ref{DefiningSets}), and 
suppose  $D_{0}=\{d_1, \ldots, d_n\}$. Then a linear code    involving $D_{0}$ is defined by
\be\nn\label{LinearCode}
\mathcal {C}_{D_{0}}=\{  \bold c_\omega= (\Tr{^m}(\omega d_1), \ldots, \Tr{^m}(\omega d_n))  \colon\, \omega \in\mathbb{F}_{p^m}\},
\ee
whose length $n=p^{m-1}-1$ and dimension at most $m$. For every $\omega\in \F_{p^m}^{\star}$,  the Hamming weights $\wt(\bold c_\omega)$  can be derived from Lemma \ref{Lemma0}, and  the weight distribution is determined by Lemmas \ref{SupportLemma} and \ref{Lemmag}.

\begin{theorem}\label{TheoremD0}
Let $f\in \textit{WRPB}$ and  $D_{0}$ be defined as in (\ref{DefiningSets}).  When $m+s$ is even, the set  $\sC_{D_{0}}$ is a three-weight linear  
$[ p^{m-1}-1,m]$ code with  weight distribution  listed in  Table \ref{Table0}.  
\end{theorem} 
\begin{proof} 
From the definition of  the code, for every $\omega \in\F_{p^m}^{\star}$, the Hamming weight of  nonzero codeword $\bold c_\omega$  is given as
$
\wt(\bold c_\omega)=\# D_{0}-\sN_{0}(\omega)+1,$
where $\sN_{0}(\omega)$ is defined as in Lemma \ref{Lemma0}. 
We can then find the Hamming weights by using Lemma \ref{Lemma0}.
 For every  $ \omega\in \sS_f,$  
\be\nn
\wt(\bold c_\omega)=
\left\{\begin{array}{ll}
(p-1)(p^{m-2} - \epsilon  (p-1)\sqrt{p^*}^{m+s-4}),  & \mbox{ if } f^{\star}(\omega)=0, \\
(p-1)(p^{m-2} +\epsilon  \sqrt{p^*}^{m+s-4}),& \mbox{ if } f^{\star}(\omega)\neq 0,
 \end{array}\right.
\ee
whose  weight distribution  is determined by Lemma \ref{Lemmag}. 
For every $\omega\in\F_{p^m}^{\star} \setminus \sS_f,$  we have
$\wt(\bold c_\omega)=(p-1)p^{m-2},$
and the number of such codewords $\bold c_\omega$ equals  $p^m-p^{m-s}-1$    by Lemma \ref{SupportLemma}. These parameters are  listed in  Table \ref{Table0}.  
Since  $\wt(\bold c_\omega)> 0$  for every $\omega\in \F_{p^m}^{\star}$,  the code $\mathcal {C}_{D_{0}}$ has $p^m$ different codewords, namely, its dimension equals $m$. The proof  is then completed. 
\end{proof}
 
 \begin{example}\label{Example0} 
For a  quadratic $1$-plateaued balanced  function $f:\F_{3^5}\to \F_{3}$, the code 
 $\sC_{D_{0}}$  is a three-weight ternary  $[80,5,48]$  code  with  weight enumerator   $1+ 60y^{48} + 161y^{54} +21y^{66}$.
This code is almost optimal  
since the best known linear codes with length $80$  and dimension $5$ has $d=53$ according to \cite{grassl2008bounds}.
\end{example}

 \begin{example}\label{Example00} 
For a  quadratic $1$-plateaued balanced  function $f:\F_{5^3}\to \F_{5}$, the code 
 $\sC_{D_{0}}$  is a three-weight  $[24,3,16]$  code  with  weight enumerator   $1+ 24y^{16} + 99y^{20} +1y^{36}$.
This code is almost optimal  by \cite{grassl2008bounds}.
\end{example}

\begin{remark} If $m+s$ is odd, then the code $\mathcal {C}_{D_{0}}$ has the same parameters of  $\mathcal {C}_{D_f}$ in \cite[Theorem 1]{IEEE}.
\end{remark}

The following theorem constructs the code $\mathcal {C}_{D_{1}}$ of the form  (\ref{LinearCodes})  involving  the defining set $D_{1}$.
We  recall that
\be\nn
\eta_0(-1)= \left\{ \begin{array}{ll}
\,\,\,\, 1 & \mbox{ if and only if } \, \, p \equiv 1\pmod 4,\\
-1 &  \mbox{ if and only if }  \,  p \equiv 3\pmod 4.
\end{array} \right.
\ee

\begin{theorem}\label{TheoremD1}  
Let $f\in \textit{WRPB}$ and $D_{1}$ be defined as in (\ref{DefiningSets}). Then, the set  $\sC_{D_{1}}$ is a three-weight linear
$\left[p^{m-1},m\right]$  code whose weight distribution  is listed in  Tables \ref{TableD1odd}, \ref{TableD1oddd} and \ref{TableD1even}.
\end{theorem}

\begin{proof} We first state that  the length of   $\sC_{D_{1}}$ is  the size of the defining set $D_{1}$.
From  its definition,  we can easily observe that  
$\wt(\bold c_ \omega)=\# D_{1}- \sN_{1}(\omega),$
for every $ \omega \in\F_{p^m}^{\star}$,  where $\sN_{1}(\omega)$ is  given in Lemma \ref{Lemma12}. This lemma is then able to compute  the Hamming weights.   Suppose that $m+s$ is odd.
For every $ \omega\in\F_{p^m}^{\star} \setminus \sS_f,$   we have
$\wt(\bold c_ \omega)=p^{m-2} (p-1)$; otherwise, 
\be\nn\begin{array}{ll}
\wt(\bold c_ \omega)&= \left\{\begin{array}{ll}
 (p-1)(p^{m-2}-  \epsilon\eta_0(-1) \sqrt{p^*}^{m+s-3}),  & \mbox{ if } f^{\star}( \omega)=0, \\
p^{m-2} (p-1)+  \epsilon  \sqrt{p^*}^{m+s-3}(\eta_0(-1)+1),& \mbox{ if } f^{\star}( \omega) \in SQ,\\
p^{m-2} (p-1) + \epsilon  \sqrt{p^*}^{m+s-3}(\eta_0(-1)-1),  & \mbox{ if }  f^{\star}( \omega) \in NSQ. \\
 \end{array}\right.
 \end{array}
\ee 
The weight distribution is determined by Lemmas \ref{SupportLemma} and \ref{Lemmag}. Note that these parameters are listed in Tables \ref{TableD1odd} and  \ref{TableD1oddd} when  $p \equiv 1\pmod 4$ and  $p \equiv 3\pmod 4$, respectively.
  When  $m+s$ is even, with the same method, we can find the corresponding parameters listed in Table \ref{TableD1even}, thereby completing the proof. 
\end{proof}

 \begin{example}\label{Example1}
For a  quadratic $1$-plateaued balanced  function $f:\F_{5^3}\to \F_{5}$, the code 
$\sC_{D_{1}}$  is a three-weight   $[25,3,16]$  code  with  weight enumerator   $1+ 13y^{16} + 99y^{20} +12y^{26}$.
This code is almost optimal  by \cite{grassl2008bounds}.
\end{example}

The following theorem constructs the code $\mathcal {C}_{D_{(0,1)}}$ of the form  (\ref{LinearCodes})  involving the  set $D_{(0,1)}$.
\begin{theorem}\label{Theorem01}
Let  $f\in \textit{WRPB}$ and $D_{(0,1)}$ be defined as in (\ref{DefiningSets}). Then, the set $\sC_{D_{(0,1)}}$ is a four-weight linear  
$[2p^{m-1}-1,m]$ code with  weight distribution listed in  Tables \ref{TableD01odd} and \ref{TableD01even}. 
\end{theorem}

\begin{proof} We proceed the proof only when $m+s$ is odd.
For every $\omega \in\F_{p^m}^{\star}$,  the Hamming weight
$\wt(\bold c_\omega)=\# D_{(0,1)}- \sN_{0}(\omega)-\sN_{1}(\omega)+1$
can be found by using  Lemmas \ref{Lemma0} and \ref{Lemma12}. 
For every $\omega\in\F_{p^m}^{\star} \setminus \sS_f,$   we have
$\wt(\bold c_\omega)=2(p-1)p^{m-2}$; otherwise,
\be\nn\begin{array}{ll}
\wt(\bold c_\omega)&= \left\{\begin{array}{ll}
2(p-1)p^{m-2}-\epsilon\eta_0(-1)(p-1) \sqrt{p^*}^{m+s-3},  & \mbox{ if } f^{\star}(\omega)=0, \\
2(p-1)p^{m-2} -  \epsilon (p-2-\eta_0(-1)) \sqrt{p^*}^{m+s-3}  & \mbox{ if } f^{\star}(\omega) \in SQ,\\
2(p-1)p^{m-2} +  \epsilon (p-2+\eta_0(-1)) \sqrt{p^*}^{m+s-3}   & \mbox{ if }  f^{\star}(\omega) \in NSQ. \\
 \end{array}\right.
 \end{array}
\ee 
The weight distribution  is determined by Lemmas \ref{SupportLemma} and  \ref{Lemmag}.  When  $m+s$ is even, with the same method, we can clearly find the corresponding parameters listed in Table  \ref{TableD01even}. Hence the proof is complete.
\end{proof}
We point out that the code  $\sC_{D_{(0,1)}}$ in Theorem \ref{Theorem01}  is the three-weight ternary code when $p=3$. As an example, we give the following code.

 \begin{example}\label{Example01}
For a  quadratic $1$-plateaued balanced  function $f:\F_{3^4}\to \F_{3}$, the code 
$\sC_{D_{(0,1)}}$  is a three-weight ternary  $[53,4,30]$  code  with  weight enumerator   $1+ 9y^{30} + 65y^{36} +6y^{42}$.
This code is almost optimal  
since the best known  linear code has $d=35$ by \cite{grassl2008bounds}.
\end{example}

 We next use the defining set  $D_{2}$ from (\ref{DefiningSets})  to define the  code    $\mathcal {C}_{D_{2}}$  of the form  (\ref{LinearCodes}), whose parameters are collected in the following theorem.
\begin{theorem}\label{TheoremD2}
Let  $f\in \textit{WRPB}$ and $D_{2}$ be defined as in (\ref{DefiningSets}). When $m+s$ is odd, the set $\mathcal {C}_{D_{2}}$ is a   three-weight linear
$\left[p^{m-1},m\right]$ code with weight distribution given in  Tables \ref{TableD2} and  \ref{TableD2odd}. 
\end{theorem}
\begin{proof}
For every $ \omega \in\F_{p^m}^{\star}$, we have
$\wt(\bold c_ \omega)=\#  D_{2}- \sN_{2}(\omega)$, where $\sN_{2}(\omega)$ is given in  Lemma \ref{Lemma12}. It  follows then from Lemma \ref{Lemma12} that   we have 
$\wt(\bold c_ \omega)=(p-1)p^{m-2}$  if $ \omega\in\F_{p^m}^{\star} \setminus \sS_f;$  otherwise,
\be\nn\begin{array}{ll}
\wt(\bold c_ \omega)=\left\{\begin{array}{ll}
(p-1)( p^{m-2} +  \epsilon\eta_0(-1)  \sqrt{p^*}^{m+s-3}),  & \mbox{ if } f^{\star}( \omega)=0, \\
p^{m-2} (p-1) -  \epsilon    \sqrt{p^*}^{m+s-3}(\eta_0(-1)-1),  & \mbox{ if } f^{\star}( \omega) \in SQ, \\
p^{m-2} (p-1) - \epsilon  \sqrt{p^*}^{m+s-3}(\eta_0(-1)+1), 
& \mbox{ if } f^{\star}( \omega) \in NSQ.
 \end{array}\right.
 \end{array}
\ee 
The weight distribution is determined by  Lemmas \ref{SupportLemma} and  \ref{Lemmag}. 
The parameters are listed in Tables \ref{TableD2} and  \ref{TableD2odd} when  $p \equiv 1\pmod 4$ and  $p \equiv 3\pmod 4$, respectively, thereby completing the proof.
\end{proof}
\begin{remark} If $m+s$ is even,  then the code $\mathcal {C}_{D_{2}}$ has  the same parameters  of $\mathcal {C}_{D_{1}}$ in Theorem \ref{TheoremD1}.
\end{remark}

 We further study the code $\mathcal {C}_{D_{(0,2)}}$ of the form  (\ref{LinearCodes})  involving  $D_{(0,2)}$.   The following theorem collects its  parameters.
\begin{theorem}\label{Theorem02}
Let  $f\in \textit{WRPB}$ and $D_{(0,2)}$ be defined as in (\ref{DefiningSets}).  When $m+s$ is odd, the set $\sC_{D_{(0,2)}}$ is a four-weight linear  
$[2p^{m-1}-1,m]$ code with weight distribution given  in  Table \ref{Table02}.
\end{theorem}

\begin{proof} 
 For every $\omega \in\F_{p^m}^{\star}$, the Hamming weight 
$\wt(\bold c_\omega)=\# D_{(0,2)}- \sN_{0}(\omega)-\sN_{2}(\omega)+1$
can be found by considering   Lemmas \ref{Lemma0} and \ref{Lemma12}.
 Then  we   have $\wt(\bold c_\omega)=2(p-1)p^{m-2}$ if  $\omega\in\F_{p^m}^{\star} \setminus \sS_f$; otherwise, 
\be\nn\begin{array}{ll}
\wt(\bold c_\omega)&= \left\{\begin{array}{ll}
2(p-1)p^{m-2} + \epsilon\eta_0(-1)(p-1) \sqrt{p^*}^{m+s-3},  & \mbox{ if } f^{\star}(\omega)=0, \\
2(p-1)p^{m-2} -  \epsilon (p-2+\eta_0(-1)) \sqrt{p^*}^{m+s-3}  & \mbox{ if } f^{\star}(\omega) \in SQ,\\
2(p-1)p^{m-2} +  \epsilon (p-2 - \eta_0(-1)) \sqrt{p^*}^{m+s-3}   & \mbox{ if }  f^{\star}(\omega) \in NSQ, \\
 \end{array}\right.
 \end{array}
\ee 
whose weight distributions  are, respectively, determined by Lemmas \ref{SupportLemma} and \ref{Lemmag}. Hence the proof is completed.  
\end{proof}

The following theorem  constructs the code $\sC_{D_{(1,2)}}$ of the form  (\ref{LinearCodes})  involving  $D_{(1,2)}$.
\begin{theorem}\label{Theorem12}
Let $f\in \textit{WRPB}$ and $D_{(1,2)}$ be defined as in (\ref{DefiningSets}).  Then, the set $\mathcal {C}_{D_{(1,2)}}$ is a   three-weight  linear  
$[2p^{m-1},m]$ code with weight distribution given in  Tables \ref{table12odd} and \ref{table12}.
\end{theorem}
\begin{proof}
For every $\omega \in\F_{p^m}^{\star}$,  the Hamming weight $\wt(\bold c_\omega)=\#  D_{(1,2)}- \sN_{1}(\omega)-\sN_{2}(\omega)$ can be  computed by using Lemma \ref{Lemma12}.
Suppose that $m+s$ is odd. 
We then have  $\wt(\bold c_\omega)=2(p-1)p^{m-2}$ if $\omega\in\F_{p^m}^{\star} \setminus \sS_f$; otherwise,  
\be\nn\begin{array}{ll}
\wt(\bold c_\omega)=\left\{\begin{array}{ll}
2(p-1)p^{m-2},  & \mbox{ if } f^{\star}(\omega)=0, \\
2(p-1)p^{m-2}+2\epsilon  \sqrt{p^*}^{m+s-3}& \mbox{ if } f^{\star}(\omega) \in SQ, \\
2(p-1)p^{m-2} - 2\epsilon  \sqrt{p^*}^{m+s-3}, & \mbox{ if } f^{\star}(\omega) \in NSQ,
 \end{array}\right.
 \end{array}
\ee 
whose weight distributions   are, respectively, determined  by Lemmas  \ref{SupportLemma} and \ref{Lemmag}.  We do not proceed the case of $m+s$ is even since the corresponding parameters listed in Table \ref{table12} can be easily  obtained  with the same method.  The proof hence is   complete. 
\end{proof}

 The following theorem considers the code $\mathcal {C}_{D_{sq}}$ of the form  (\ref{LinearCodes})  involving  $D_{sq}$.
\begin{theorem}\label{TheoremSQ}
Let $f\in \textit{WRPB}$ and $D_{sq}$ be defined as in (\ref{DefiningSets}).  Then, the set  $\sC_{D_{sq}}$ is a three-weight linear  
$[p^{m-1} (p-1)/2,m]$ code whose weight distribution is given in  Tables \ref{TableSQodd},  \ref{TableSQoddd} and  \ref{TableSQeven}. 
\end{theorem}

\begin{proof} For every $\omega \in\F_{p^m}^{\star}$,
the Hamming weight $\wt(\bold c_\omega)=\# D_{sq}- \sN_{sq}(\omega)$ 
 follows  from  Lemma \ref{LemmaSQ}. 
We proceed the proof only when $m+s$ is odd.
If $\omega\in\F_{p^m}^{\star} \setminus \sS_f$, then  we have $\wt(\bold c_\omega)=p^{m-2} (p-1)^2/2$; otherwise, 
\be\nn\begin{array}{ll}
\wt(\bold c_\omega)&= \left\{\begin{array}{ll}
\frac{1}{2}(p-1)^2(p^{m-2}-  \epsilon\eta_0(-1) \sqrt{p^*}^{m+s-3}),  & \mbox{ if } f^{\star}(\omega)=0, \\
p^{m-2} (p-1)^2/2+  \epsilon  \frac{1}{2}(p-1) \sqrt{p^*}^{m+s-3}(\eta_0(-1)+1),& \mbox{ if } f^{\star}(\omega) \in SQ,\\
p^{m-2} (p-1)^2/2 + \epsilon \frac{1}{2}(p-1) \sqrt{p^*}^{m+s-3}(\eta_0(-1)-1),  & \mbox{ if }  f^{\star}(\omega) \in NSQ.\\
 \end{array}\right.
 \end{array}
\ee 
The weight distribution can be determined by using  Lemmas \ref{SupportLemma} and \ref{Lemmag}.
The parameters are listed in Tables  \ref{TableSQodd} and  \ref{TableSQoddd} when  $p \equiv 1\pmod 4$ and  $p \equiv 3\pmod 4$, respectively.
When $m+s$ is even,  we immediately obtain the corresponding parameters listed in Table \ref{TableSQeven},  completing the proof.  
\end{proof}

We  use the defining set $D_{(sq,0)}$  from (\ref{DefiningSets}) to define the code $\mathcal {C}_{D_{(sq,0)}}$  of the form  (\ref{LinearCodes}), whose parameters are collected in the following theorem.

\begin{theorem}\label{TheoremSQ0}
Let  $f\in \textit{WRPB}$ and $D_{(sq,0)}$ be defined as in (\ref{DefiningSets}). Then, the set  $\sC_{D_{(sq,0)}}$ is a   three-weight linear 
$\left[p^{m-1}(p+1)/2-1,m\right]$ code whose weight distribution  is documented in  Tables \ref{tableSQ0}, \ref{tableSQ00} and \ref{tableSQ0even}.
\end{theorem}

\begin{proof}
For every $ \omega \in\F_{p^m}^{\star}$,  $\wt(\bold c_ \omega)=\# D_{(sq,0)}- \sN_{(sq,0)}( \omega)+1$
follows  from Lemma \ref{LemmaSQ}.    When $m+s$ is odd,   we have
$\wt(\bold c_ \omega)=p^{m-2}(p^2-1)/2$ if $ \omega\in\F_{p^m}^{\star} \setminus \sS_f$; otherwise, 
\be\nn\begin{array}{ll}
\wt(\bold c_ \omega)&=\left\{\begin{array}{ll}
p^{m-2}(p^2-1)/2-\epsilon\eta_0(-1)\frac{1}{2}(p-1)^2 \sqrt{p^*}^{m+s-3},  & \mbox{ if } f^{\star}( \omega)=0, \\
p^{m-2}(p^2-1)/2+ \epsilon  \frac{1}{2}(p-1)\sqrt{p^*}^{m+s-3}(\eta_0(-1)-1),& \mbox{ if } f^{\star}( \omega) \in SQ,\\
p^{m-2}(p^2-1)/2+ \epsilon  \frac{1}{2}(p-1) \sqrt{p^*}^{m+s-3}(\eta_0(-1)+1),  & \mbox{ if }  f^{\star}( \omega) \in NSQ, \\
 \end{array}\right.
 \end{array}
\ee
which are listed in  Tables \ref{tableSQ0} and \ref{tableSQ00}  when  $p \equiv 1\pmod 4$ and  $p \equiv 3\pmod 4$, respectively. 
When $m+s$ is even, it is easy to get the corresponding parameters listed in Table \ref{tableSQ0even}. Finally,  the weight distribution is determined by using  Lemmas \ref{SupportLemma} and \ref{Lemmag}, completing the proof.  
\end{proof}
We below introduce the parameters of the code $\mathcal {C}_{D_{nsq}}$ of the form  (\ref{LinearCodes})  involving  $D_{nsq}$.

\begin{theorem}\label{TheoremNSQ}
Let $f\in \textit{WRPB}$ and $D_{nsq}$ be defined as in (\ref{DefiningSets}).  When $m+s$ is odd, the set  $\mathcal {C}_{D_{nsq}}$ is a   three-weight linear  
$[p^{m-1} (p-1)/2,m]$ code with weight distribution  given  in  Tables \ref{tableNSQ} and \ref{tableNSQ3}. 
\end{theorem}
\begin{proof}
  For every $\omega \in\F_{p^m}^{\star}$, the Hamming weight 
$\wt(\bold c_\omega)=\#  D_{nsq}- \sN_{nsq}(\omega)$ follows from  
Lemma \ref{LemmaSQ}. 
We then have
$\wt(\bold c_\omega)=\frac{1}{2} (p-1)^2p^{m-2} $ if $\omega\in\F_{p^m}^{\star} \setminus \sS_f;$ otherwise, 
\be\nn\begin{array}{ll}
\wt(\bold c_\omega)=\left\{\begin{array}{ll}
\frac{1}{2}(p-1)^2( p^{m-2} +  \epsilon\eta_0(-1)  \sqrt{p^*}^{m+s-3}),  & \mbox{ if } f^{\star}(\omega)=0, \\
p^{m-2} (p-1)^2/2 -  \epsilon \frac{1}{2}(p-1)   \sqrt{p^*}^{m+s-3}(\eta_0(-1)-1),  & \mbox{ if } f^{\star}(\omega) \in SQ, \\
p^{m-2} (p-1)^2/2 - \epsilon\frac{1}{2}(p-1)   \sqrt{p^*}^{m+s-3}(\eta_0(-1)+1), 
& \mbox{ if } f^{\star}(\omega) \in NSQ.
 \end{array}\right.
 \end{array}
\ee 
The weight distribution  is determined with the help of  Lemmas \ref{SupportLemma} and \ref{Lemmag}, completing   the proof.
\end{proof}

\begin{remark} If $m+s$ is even, then the code $\mathcal {C}_{D_{nsq}}$ has  the same parameters of $\mathcal {C}_{D_{sq}}$ in Theorem \ref{TheoremSQ}.
\end{remark}
 We finally use the  defining set $D_{(nsq,0)}$  from (\ref{DefiningSets}) to define the  code $\mathcal {C}_{D_{(nsq,0)}}$  of the form (\ref{LinearCodes}),
whose parameters are listed in the following theorem.

 \begin{theorem}\label{TheoremNSQ0}
Let $f\in \textit{WRPB}$ and $D_{(nsq,0)}$ be defined as in (\ref{DefiningSets}). When $m+s$ is odd, the set  $\sC_{D_{(nsq,0)}}$ is a   three-weight $\left[p^{m-1}(p+1)/2-1,m\right]$ code with  weight distribution listed in  Tables \ref{tableNSQ0} and \ref{tableNSQ03}.  
\end{theorem}
\begin{proof}
  For every $ \omega \in\F_{p^m}^{\star}$,  $\wt(\bold c_ \omega)=\# D_{(nsq,0)}- \sN_{(nsq,0)}( \omega)+1$
follows from  Lemma \ref{LemmaSQ}. 
Then we have
$\wt(\bold c_ \omega)=p^{m-2}(p^2-1)/2$ if $ \omega\in\F_{p^m}^{\star} \setminus \sS_f$; otherwise, 
\be\nn\begin{array}{ll}
\wt(\bold c_ \omega)&=\left\{\begin{array}{ll}
p^{m-2}(p^2-1)/2 + \epsilon\eta_0(-1)\frac{1}{2}(p-1)^2 \sqrt{p^*}^{m+s-3},  & \mbox{ if } f^{\star}( \omega)=0, \\
p^{m-2}(p^2-1)/2- \epsilon  \frac{1}{2}(p-1)\sqrt{p^*}^{m+s-3}(\eta_0(-1)+1),& \mbox{ if } f^{\star}( \omega) \in SQ,\\
p^{m-2}(p^2-1)/2- \epsilon  \frac{1}{2}(p-1) \sqrt{p^*}^{m+s-3}(\eta_0(-1)-1),  & \mbox{ if }  f^{\star}( \omega) \in NSQ. \\
 \end{array}\right.
\end{array}
\ee
The weight distribution  is determined by Lemmas  \ref{SupportLemma} and  \ref{Lemmag}, thereby completing the proof.
\end{proof}
\begin{remark} If $m+s$ is even, then $\mathcal {C}_{D_{(nsq,0)}}$ has the same parameters of  $\mathcal {C}_{D_{(sq,0)}}$  in Theorem \ref{TheoremSQ0}.
\end{remark}
\begin{remark} 
The length and dimension of each constructed  code follows, respectively,  from (\ref{length})  and its weight distribution.  
\end{remark}

We end this section by   proposing  a shorter linear code, which is called \textit{punctured code},  for the code $\sC_{D_{0}}$ defined as in  (\ref{LinearCodes}).
For $f\in \textit{WRPB}$, we have that 
$f(x)=0$ if and only if $f(ax)=0$  for every $x\in\F_{p^m}$ and  $a\in\F_p^{\star}$.
Then we select a subset $\overline{D}_{0}$ of the defining set $D_{0}$ of  $\sC_{D_{0}}$  such that $\bigcup_{a\in \F_p^{\star}}a\overline{D}_{0}$ is a partition of $D_{0}$, namely,
\be\label{Punctured}
D_{0}=\F_p^{\star}\overline{D}_{0}=\{a \bar{d}  \colon\, a\in\F_p^{\star} \mbox{  and }\bar{d} \in  \overline{D}_{0}\},
\ee
where we have $\frac{\bar{d_1}}{\bar{d_2}}\notin\F_p^{\star}$ for every  $\bar{d_1},\bar{d_2}\in \overline{D}_{0}$.
Notice that for every $\omega\in\F_{p^m}^{\star}$,  we have
$ \#\{x\in D_{0}  \colon\, f(x)=0 \mbox{ and } \Tr^m(\omega x)=0\}=(p-1) \#\{x\in \overline{D}_{0}  \colon\, f(x)=0 \mbox{ and } \Tr^m(\omega x)=0\}.$
Hence,  the code $\sC_{D_{0}}$ can be punctured into a shorter linear code $\sC_{\overline{D}_{0}}$ involving the defining set $\overline{D}_{0}$.  
This method  decreases  the minimum Hamming distance and length of the original code while its dimension does not change. The  punctured codes then  may  be optimal codes by  \cite{grassl2008bounds}.   The parameters of the punctured code $\sC_{\overline{D}_{0}}$  are collected in the following corollary.

\begin{corollary}\label{CorollaryPunc}
The punctured code $\sC_{\overline{D}_{0}}$  of the   code $\sC_{D_{0}}$ in Theorem \ref{TheoremD0} is a three-weight linear  $ [(p^{m-1}-1)/(p-1),m]$ code with  weight distribution listed in  Table \ref{tablePunc}. 
\end{corollary}

 \begin{example}\label{Example0Punc}
 The punctured version $\sC_{\overline{D}_{0}}$ of Example  \ref{Example0}  is a three-weight  ternary   $[40,5,24]$  code with  weight enumerator   $1+ 60y^{24} + 161y^{27} +21y^{33}$.
This code is  optimal by \cite{grassl2008bounds}.  
\end{example}

 \begin{example}\label{Example00Punc}
 The punctured version $\sC_{\overline{D}_{0}}$ of Example  \ref{Example00}  is a three-weight    $[12,3,8]$  code  over $\F_5$ with  weight enumerator  $1+ 24y^{8} + 99y^{10} +1y^{18}$.
This code is  optimal by \cite{grassl2008bounds}.   
\end{example}

\begin{remark}
The projective three-weight punctured code of Corollary \ref{CorollaryPunc}  provides an association scheme given in \cite{calderbank1984three}.
\end{remark}
When $p=3$, the code $\sC_{D_{(1,2)}}$ in  Theorem \ref{Theorem12}  can be punctured into a shorter linear code $\sC_{\overline{D}_{(1,2)}}$ involving the defining set $\overline{D}_{(1,2)}$ defined as in (\ref{Punctured})
\begin{corollary}\label{Corollary12}
The punctured code  $\sC_{\overline{D}_{(1,2)}}$ of Theorem \ref{Theorem12} is a   three-weight ternary  linear   $[3^{m-1},m]$ code over $\F_3$ with weight distribution given in  Tables \ref{table12oddpunc} and \ref{table12punc},  when $p=3$.  
\end{corollary}
 \begin{example}\label{Example12}
For a  quadratic $1$-plateaued balanced  function $f:\F_{3^3}\to \F_{3}$, the punctured code 
$\sC_{\overline{D}_{(1,2)}}$  is a three-weight ternary  $[9,3,5]$  code  with  weight enumerator   $1+ 4y^{5} + 17y^{6} +5y^{8}$.
This code is almost optimal owing to the Griesmer bound. 
\end{example}

\section{Application of  the constructed  codes in secret sharing}\label{SectionSSS}
In this section, we study an application of the constructed  codes in secret sharing.
 
We first recall the the following sufficient  condition for minimal   codes  introduced by Ashikhmin et al. \cite{ashikhmin1998minimal}.
Let $\sC$ be a linear code  over $\F_p$ and denote by $w_{\min}$ and $w_{\max}$   the minimum and maximum  Hamming weights of its nonzero codewords, respectively. Then,  $\sC$  is  a minimal code if
\be\label{minimality}
\frac{p-1}{p}<\frac{w_{\min}}{w_{\max}}.
\ee
With the help of the condition  in (\ref{minimality}), we observe that our  codes  are minimal, thereby they have an interesting application in secret sharing. 

We recall that  $f:\F_{p^m}\to\F_{p}$ is an $s$-plateaued balanced function from the class $\textit{WRPB}$, where  $s\in\{1,\ldots,m\}$.
We now  see that the code $\sC_{D_{0}}$ in Theorem \ref{TheoremD0} is minimal for $s\in\{1,\ldots,m-4\}$, and similarly the others can be easily seen by putting  a necessary bound on  $s\in\{1,\ldots,m\}$. We  provide the parameters of our  minimal codes in the following propositions. 
 Suppose that  the sign $ \epsilon \eta_0^{(m+s)/2}(-1)$ and $\epsilon \eta_0^{(m+s-3)/2}(-1)$ is, respectively,  denoted by   $\epsilon_0$ and $\epsilon_1$.

\begin{proposition} 
 The  code $\sC_{D_{0}}$ in Theorem \ref{TheoremD0} is  minimal over $\F_p$  for $1\leq s\leq m-4$  with parameters 
$[ p^{m-1}-1,m, (p-1)(p^{m-2}-(p-1)\sqrt{p}^{m+s-4})]$ if  $\epsilon_0=1$, and 
$[p^{m-1}-1,m,  (p-1) (p^{m-2}-\sqrt{p}^{m+s-4})]$; otherwise. 
\end{proposition}

\begin{proof} If $\epsilon_0=1$, then we have
$w_{\min}=(p-1)(p^{m-2}-(p-1)\sqrt{p}^{m+s-4})  \mbox{ and } w_{\max}=(p-1)(p^{m-2}+\sqrt{p}^{m+s-4}).$
 Otherwise, 
$w_{\min}= (p-1)(p^{m-2}-\sqrt{p}^{m+s-4})$ and 
$w_{\max}=(p-1)(p^{m-2}+(p-1)\sqrt{p}^{m+s-4}).$  
For both cases, the  sufficient condition in (\ref{minimality}) is satisfied when  $1\leq s\leq m-4$.
Hence, this observation completes the proof.
\end{proof} 

\begin{proposition}\label{Theorem10}
 The code $\sC_{D_{1}}$ in Theorem \ref{TheoremD1} is   minimal  over $\F_p$.
When $m+s$ is  odd   with  $1\leq s\leq m-5$, we have
$[p^{m-1}, m, (p-1)(p^{m-2}-\sqrt{p}^{m+s-3})]$ if  $\epsilon=1$  when $p\equiv 1 \pmod 4$  and $\epsilon_1=-1$ when $p\equiv 3 \pmod 4$. Otherwise,
$[p^{m-1}, m,   (p-1)p^{m-2}  - 2\sqrt{p}^{m+s-3}]$.
When $m+s$ is  even   with  $1\leq s\leq m-4$,  we have
$[p^{m-1}, m, (p-1)p^{m-2}-(p+1)\sqrt{p}^{m+s-4}] $ if $  \epsilon_0=1$, and 
$[p^{m-1}, m,(p-1) (p^{m-2}- \sqrt{p}^{m+s-4}) ]$;  otherwise.
\end{proposition}

\begin{proposition} 
The code $\sC_{D_{(0,1)}}$ in Theorem \ref{Theorem01}  is  minimal over $\F_p$. 
When  $m+s$ is   odd   with  $1\leq s\leq m-3$, we have
$[2p^{m-1}-1, m, (p-1)(2p^{m-2}-\sqrt{p}^{m+s-3})]$.
 When  $m+s$ is  even   with  $1\leq s\leq m-4$, we have
$[2p^{m-1}-1, m, (p-1)(2p^{m-2}- (p-2)\sqrt{p}^{m+s-4})]$ if $ \epsilon_0=1$, and 
$[2p^{m-1}-1,m,2(p-1)(p^{m-2} - \sqrt{p}^{m+s-4})]$; otherwise.
\end{proposition}
 
\begin{proposition}
 The code $\sC_{D_{2}}$  in Theorem \ref{TheoremD2} is  minimal  over $\F_p$ for $1\leq s\leq m-5$.  If $\epsilon=1$ when   $p\equiv 1 \pmod 4$  and $\epsilon_1=-1$ when $p\equiv 3 \pmod 4$, then we have
$[p^{m-1}, m,  p^{m-2} (p-1)  - 2\sqrt{p}^{m+s-3}]$, and  
$[p^{m-1}, m, (p-1)(p^{m-2}-\sqrt{p}^{m+s-3})]$; otherwise.
\end{proposition}

\begin{proposition} 
 The code $\sC_{D_{(0,2)}}$ in Theorem \ref{Theorem02} is  the minimal $[2p^{m-1}-1, m, (p-1)(2p^{m-2}-\sqrt{p}^{m+s-3}) ]$ code over $\F_p$ for $1\leq s\leq m-3$.
\end{proposition}

\begin{proposition} 
The code $\sC_{D_{(1,2)}}$ in Theorem \ref{Theorem12} is  minimal over $\F_p$. 
When $m+s$ is odd   with  $1\leq s\leq m-3$, we have 
$[2p^{m-1}, m,  2(p-1)p^{m-2} -2  \sqrt{p}^{m+s-3}]$. 
When  $m+s$ is  even, we have
$[2p^{m-1}, m, 2(p-1)p^{m-2}- 2\sqrt{p}^{m+s-4}]$
if  $\epsilon_0=1$;  otherwise,
$[2p^{m-1},m, 2(p-1)(p^{m-2}-\sqrt{p}^{m+s-4})]$,  for $1\leq s\leq m-2$ and  $1\leq s\leq m-4$, respectively.
\end{proposition} 

\begin{proposition}
The code  $\sC_{D_{sq}}$ in Theorem \ref{TheoremSQ} is minimal over $\F_p$. 
When  $m+s$ is odd   with  $1\leq s\leq m-5$, if 
 $\epsilon=1$ when   $p\equiv 1 \pmod 4$  and $\epsilon_1=-1$ when $p\equiv 3 \pmod 4$, we have
$[p^{m-1} (p-1)/2, m,(p-1)^2(p^{m-2}-\sqrt{p}^{m+s-3})/2]$;
 otherwise,
$[p^{m-1} (p-1)/2, m,  (p-1) ( p^{m-2}(p-1)/2  -\sqrt{p}^{m+s-3})]$. 
When $m+s$ is  even   with  $1\leq s\leq m-4$, we have
$[p^{m-1} (p-1)/2, m, \frac{1}{2}(p-1)((p-1)p^{m-2}-(p+1)\sqrt{p}^{m+s-4})]$ if $\epsilon_0=1,$ and
$[p^{m-1} (p-1)/2, m,  \frac{1}{2}(p-1)^2(p^{m-2}- \sqrt{p}^{m+s-4})]$;  otherwise.
\end{proposition}
 
\begin{proposition}
The code $\sC_{D_{(sq,0)}}$ in Theorem \ref{TheoremSQ0} is  minimal over $\F_p$.
When   $m+s$ is  odd   with  $1\leq s\leq m-3$,  
if  $\epsilon=1$ when   $p\equiv 1 \pmod 4$  and $\epsilon_1=-1$ when $p\equiv 3 \pmod 4$, we have
$[p^{m-1}(p+1)/2-1, m, \frac{1}{2}(p-1)( p^{m-2}(p+1)- (p-1) \sqrt{p}^{m+s-3})]$; otherwise,
$[p^{m-1}(p+1)/2-1, m, p^{m-2}(p^2-1)/2-  (p-1) \sqrt{p}^{m+s-3}]$.
 When $m+s$ is  even   with  $1\leq s\leq m-4$,  we have 
$[p^{m-1}(p+1)/2-1, m,\frac{1}{2}(p-1)((p+1)p^{m-2} - (p-1)\sqrt{p}^{m+s-4})]$ if $\epsilon_0=1$,
and $[p^{m-1}(p+1)/2-1, m, \frac{1}{2}(p^2-1) (p^{m-2} - \sqrt{p}^{m+s-4})]$; otherwise.  
\end{proposition}
\begin{proposition}
The code $\sC_{D_{nsq}}$ in Theorem \ref{TheoremNSQ} is  minimal over $\F_p$ for   $1\leq s\leq m-5$. If  $\epsilon=1$ when   $p\equiv 1 \pmod 4$  and $\epsilon_1=-1$ when $p\equiv 3 \pmod 4$, then we have 
$[p^{m-1} (p-1)/2, m, (p-1)( p^{m-2} (p-1)/2  - \sqrt{p}^{m+s-3})]$; otherwise, 
$[p^{m-1} (p-1)/2, m, \frac{1}{2}(p-1)^2(p^{m-2}- \sqrt{p}^{m+s-3})]$.
\end{proposition}

 \begin{proposition}  
The code $\sC_{D_{(nsq,0)}}$ in Theorem \ref{TheoremNSQ0} is  minimal over $\F_p$ for $1\leq s\leq m-3$.
If  $\epsilon=1$ when   $p\equiv 1 \pmod 4$  and $\epsilon_1=-1$ when $p\equiv 3 \pmod 4$, then we have 
$[p^{m-1}(p+1)/2-1, m, (p-1)(p^{m-2}(p+1)/2-  \sqrt{p}^{m+s-3})]$; otherwise,  
$[p^{m-1}(p+1)/2-1, m, \frac{1}{2}(p-1)(p^{m-2}(p+1)- (p-1)   \sqrt{p}^{m+s-3})]$. 
\end{proposition}
As the constructed codes are all minimal codes,  secret sharing schemes based on their dual codes  have high democracy  introduced  in the following theorem.
\begin{theorem}\cite{carlet2005linear,ding2003covering}\label{Structure}
Let   $\sC$ be  a minimal  linear  $[n,k,d]$ code over $\F_p$ with the generator matrix
$G=[\bold g_0,\bold g_1,\ldots,\bold g_{n-1}]$, and let $d^\perp$ represent the minimum Hamming distance  of its dual code $\sC^\perp$.  Then  in the SSS based on $\sC^\perp$, the number of participants equals $n-1$, and the number of   minimal access sets equals $p^{k-1}$. 
\begin{itemize}
\item For $d^\perp=2$,    if $\bold g_i$, $1\leq i\leq n-1$, is a multiple of $\bold g_0$, then a participant $P_i$ is  in every minimal access set; else, 
 $P_i$ is in $(p-1)p^{k-2}$  minimal access sets.
\item For $d^\perp\geq 3$, for each fixed $1\leq l \leq \min\{k-1,d^\perp-2\}$, every set of $l$ participants is involved in $(p-1)^lp^{k-(l+1)}$  minimal access sets.
\end{itemize}
\end{theorem}
To describe the access structures of SSS based on the dual codes of our minimal  codes,  we  are first  interested in the minimum Hamming distance of the dual code.
 By  \textit{the MacWilliams identity} (F. J. MacWilliams, 1963),
the  weight enumerator  (and hence minimum Hamming distance) of the dual code  is obtained from that of the original code.
\begin{theorem}\cite[Theorem 3.5.3]{lint1999introduction}\label{Mac}
Let $\sC$ be a linear $[n,k]$ code over $\F_q$ with weight enumerator $A(z)$.  The weight enumerator of $\sC^\perp$ is denoted by $A^\perp(z)$. Then
\be\nn
 A^\perp(z)=q^{-k} (1 + (q-1)z)^n A\left( \frac{1-z}{ 1 + (q-1)z} \right).
\ee
\end{theorem}
With the help of  Theorem \ref{Mac}, one can find the weight enumerator and minimum Hamming distance of the dual code of each  code constructed in this paper. 
However, we prefer to use the following simple method in order to find $d^\perp$.
It is a well known fact that  two elements of each codeword are  dependent if and only if 
the minimum Hamming distance of the dual code  is $2$. 
In our framework, the dual code $\sC_{D_{0}}^\perp$ of Theorem \ref{TheoremD0} has   $d^\perp=2$ if and only if  for any two different elements $d_i,d_j\in D_{0}$  and  two  elements  $c_i,c_j\in\F_p^{\star}$, we have 
$$
c_i\Tr{^n}(x d_i)+c_j \Tr{^n}(x d_j)=0
$$
for every $x\in\F_{p^m}$,  which holds when $d_j=-d_i$ and $c_i=c_j=1$.   This result confirms   $d^\perp=2$.
With the same reason, the dual codes of Theorems \ref{TheoremD1},  \ref{Theorem01}, \ref{TheoremD2},   \ref{Theorem02},  \ref{Theorem12},  \ref{TheoremSQ}, \ref{TheoremSQ0}, \ref{TheoremNSQ},  \ref{TheoremNSQ0}  have   $d^\perp=2$. 
Hence, one can  give the SSS based on the dual codes of these minimal codes by considering Theorem \ref{Structure}. 
 As an example,  we deal  with the following one. 
\begin{corollary} 
Let   $\sC_{D_{(0,1)}}$ be  the minimal  $[2p^{m-1}-1, m, (p-1)(2p^{m-2}- (p-2)\sqrt{p}^{m+s-4})]$ code  in Theorem \ref{Theorem01} with 
$G=[\bold g_0,\bold g_1,\ldots,\bold g_{2p^{m-1}-2}]$. 
Then  in the SSS based on $\sC_{D_{(0,1)}}^\perp$ with  $d^\perp=2$, the number of participants is    $2p^{m-1}-2$  and the number of  minimal access sets is  $p^{m-1}$. Besides, $P_i$ must be  in all minimal access sets if $\bold g_i$, $i\neq 0$, is a multiple of $\bold g_0$; otherwise, in $(p-1)p^{m-2}$  minimal access sets.
\end{corollary}

We finally see that  $d^\perp$ of the dual code  $\sC_{\overline{D}_{0}}^\perp$ of   Corollary \ref{CorollaryPunc}   is at least $3$. From (\ref{Punctured}), we have   $D_{0}=\F_p^{\star}\overline{D}_{0}$.
Clearly,   $d^\perp=2$ if and only if for any two different elements $\bar{d_i},\bar{d_j}\in \overline{D}_{0}$ and two elements  $a_i,a_j\in\F_p^{\star}$, we have 
$\Tr{^n}(x (a_i\bar{d_i}+a_j\bar{d_j}))=0$  for every $x\in\F_{p^m}$; 
that is, $a_i\bar{d_i}+a_j\bar{d_j}=0$, which contradicts $\frac{\bar{d_i}}{\bar{d_j}}\notin \F_p^{\star}$. This observation  says  $ d^\perp\geq 3$.

\begin{corollary}
Let   $\sC_{\overline{D}_{0}}$  be  the minimal  $ [(p^{m-1}-1)/(p-1),m]$  code in Corollary \ref{CorollaryPunc}. Then  in the SSS based on $\sC_{\overline{D}_{0}}^\perp$ with $d^\perp\geq 3$, the number of participants is  $(p^{m-1}-1)/(p-1)-1$ and  the number of  minimal access sets is  $p^{m-1}$.  For each fixed $1\leq l\leq \min\{m-1,d^\perp-2\}$, every set of $l$ participants is involved in $(p-1)^lp^{m-(l+1)}$  minimal access sets.

\end{corollary}

 \begin{example}\label{Example0PuncSSS}
Let  $\sC_{\overline{D}_{0}}$ be the three-weight  ternary minimal  $[40,5,24]$  code in  Example  \ref{Example0Punc}.
Then in the SSS based on $\sC_{\overline{D}_{0}}^\perp$ with  $d^\perp\geq 3$,
the number of participants and minimal access sets is, respectively,  $39$ and  $81$.
For $l=1$, each participant is a member of  $54$  minimal access sets. 
\end{example}

\section{Conclusion}
The main objectives of the paper are twofold: to construct minimal linear codes from  functions and to give their application in secret sharing. 
To do this, we first pushed  the use of weakly regular plateaued balanced  functions over the finite fields of odd characteristic,  introduced  recently by Mesnager et al.~\cite{mesnager2017WCC,mesnager2019linear}.
  We  then obtained  several classes of three-weight  and four-weight  minimal   linear codes   from these functions with some homogeneous conditions. This paper provides the first construction of minimal linear codes with few weights from such balanced functions based on the second generic construction. To the best of our knowledge, the constructed minimal codes are  inequivalent to the known ones in the literature.
 We finally derived secret sharing schemes with nice access structures from the dual codes of our  minimal codes.

\section*{Acknowledgment} 
The author is very grateful  to the Prof. Dr. Sihem Mesnager  for her valuable scientific comments and suggestions that improved the quality of  the paper.
This research did not receive any specific grant from funding agencies in the public, commercial, or not-for-profit sectors.

%
%

\bibliographystyle{spmpsci}      


\begin{thebibliography}{10}
\expandafter\ifx\csname url\endcsname\relax
  \def\url#1{\texttt{#1}}\fi
\expandafter\ifx\csname urlprefix\endcsname\relax\def\urlprefix{URL }\fi
\expandafter\ifx\csname href\endcsname\relax
  \def\href#1#2{#2} \def\path#1{#1}\fi


\bibitem{ashikhmin1998minimal}
A.~Ashikhmin, A.~Barg, Minimal vectors in linear codes, IEEE Transactions on
  Information Theory 44~(5) ((1998)) 2010--2017.


\bibitem{calderbank1984three}
A.~Calderbank, J.~Goethals, Three-weight codes and association schemes, Philips
  J. Res 39~(4-5) (1984) 143--152.



\bibitem{carlet1998codes}
C.~Carlet, P.~Charpin, V.~Zinoviev, Codes, bent functions and permutations
  suitable for des-like cryptosystems, Designs, Codes and Cryptography 15~(2)
  (1998) 125--156.

\bibitem{carlet2005linear}
C.~Carlet, C.~Ding, J.~Yuan, Linear codes from perfect nonlinear mappings and
  their secret sharing schemes, IEEE Transactions on Information Theory 51~(6)
  (2005) 2089--2102.

\bibitem{ccesmeliouglu2012construction}
A.~{\c{C}}e{\c{s}}melio{\u{g}}lu, G.~McGuire, W.~Meidl, A construction of
  weakly and non-weakly regular bent functions, Journal of Combinatorial
  Theory, Series A 119~(2) (2012) 420--429.



\bibitem{ding2009class}
C.~Ding, A class of three-weight and four-weight codes, in: International
  Conference on Coding and Cryptology, Springer, 2009, pp. 34--42.

\bibitem{ding2015linear}
C.~Ding, Linear codes from some 2-designs, IEEE Transactions on information
  theory 61~(6) (2015) 3265--3275.



\bibitem{ding2016construction}
C.~Ding, A construction of binary linear codes from boolean functions, Discrete
  mathematics 339~(9) (2016) 2288--2303.

\bibitem{ding2007cyclotomic}
C.~Ding, H.~Niederreiter, Cyclotomic linear codes of order $3$, IEEE
  transactions on information theory 53~(6) (2007) 2274--2277.



\bibitem{ding2003covering}
C.~Ding, J.~Yuan, Covering and secret sharing with linear codes, DMTCS 2731
  (2003) 11--25.


\bibitem{ding2014binary}
K.~Ding, C.~Ding, Binary linear codes with three weights, IEEE Communications
  Letters 18~(11) (2014) 1879--1882.

\bibitem{ding2015class}
K.~Ding, C.~Ding, A class of two-weight and three-weight codes and their
  applications in secret sharing, IEEE Transactions on Information Theory
  61~(11) (2015) 5835--5842.



\bibitem{grassl2008bounds}
M.~Grassl, Bounds on the minimum distance of linear codes,
  http://www.codetables.de.

\bibitem{helleseth2006monomial}
T.~Helleseth, A.~Kholosha, Monomial and quadratic bent functions over the
  finite fields of odd characteristic, IEEE Transactions on Information Theory
  52~(5) (2006) 2018--2032.

\bibitem{hou2004solution}
X.-d. Hou, Solution to a problem of s. payne, Proceedings of the American
  Mathematical Society 132~(1) (2004) 1--6.

\bibitem{hyun2016explicit}
J.~Y. Hyun, J.~Lee, Y.~Lee, Explicit criteria for construction of plateaued
  functions, IEEE Transactions on Information Theory 62~(12) (2016) 7555--7565.



\bibitem{li2014weight}
C.~Li, N.~Li, T.~Helleseth, C.~Ding, The weight distributions of several
  classes of cyclic codes from apn monomials, IEEE transactions on information
  theory 60~(8) (2014) 4710--4721.


\bibitem{lint1999introduction}
J.~v. Lint, Introduction to coding theory, Springer, 1999.


\bibitem{mesnager2017linear}
S.~Mesnager, Linear codes with few weights from weakly regular bent functions
  based on a generic construction, Cryptography and Communications 9~(1) (2017)
  71--84.

\bibitem{mesnager2017WCC}
S.~Mesnager, F.~{\"O}zbudak, A.~S{\i}nak, A new class of three-weight linear
  codes from weakly regular plateaued functions, in: Proceedings of the Tenth
  International Workshop on Coding and Cryptography (WCC) 2017.


\bibitem{mesnager2015results}
S.~Mesnager, F.~{\"O}zbudak, A.~S{\i}nak, Results on characterizations of
  plateaued functions in arbitrary characteristic, in: International Conference
  on Cryptography and Information Security in the Balkans, Springer, 2015, pp.
  17--30.




\bibitem{mesnager2019linear}
S.~Mesnager, F.~{\"O}zbudak, A.~S{\i}nak, Linear codes from weakly regular
  plateaued functions and their secret sharing schemes, Designs, Codes and
  Cryptography 87~(2-3) (2019) 463--480.
 
\bibitem{IEEE}
S.~{Mesnager}, A.~{Sınak}, Several classes of minimal linear codes with few
  weights from weakly regular plateaued functions, IEEE Transactions on
  Information Theory 66~(4) (2020) 2296--2310.




\bibitem{tang2016linear}
C.~Tang, N.~Li, Y.~Qi, Z.~Zhou, T.~Helleseth, Linear codes with two or three
  weights from weakly regular bent functions, IEEE Transactions on Information
  Theory 62~(3) (2016) 1166--1176.

\bibitem{tang2017linear}
C.~Tang, C.~Xiang, K.~Feng, Linear codes with few weights from inhomogeneous
  quadratic functions, Designs, Codes and Cryptography 83~(3) (2017) 691--714.



\bibitem{DCCwu}
Y.~Wu, N.~Li, X.~Zeng, Linear codes with few weights from cyclotomic classes
  and weakly regular bent functions, Designs, Codes and Cryptography (2020)
  1--18.


\bibitem{zeng2012triple}
X.~Zeng, J.~Shan, L.~Hu, A triple-error-correcting cyclic code from the gold
  and kasami--welch apn power functions, Finite Fields and Their Applications
  18~(1) (2012) 70--92.


\bibitem{zheng1999plateaued}
Y.~Zheng, X.-M. Zhang, Plateaued functions, in: ICICS, Vol.~99, Springer, 1999,
  pp. 284--300.



\bibitem{zhou2016linear}
Z.~Zhou, N.~Li, C.~Fan, T.~Helleseth, Linear codes with two or three weights
  from quadratic bent functions, Designs, Codes and Cryptography 81~(2) (2016)
  283--295.


\end{thebibliography}


\section*{Appendix}
The Hamming weights and weight distributions of the  codes  constructed in Section \ref{Constructions}  are presented in Tables 1-24.

\begin{table}[!htp]
\begin{center}	
\begin{tabular}{|c|c|c|}
\hline
Hamming weight $\omega$   & Multiplicity $A_\omega$ \\
\hline
\hline
\footnotesize{$0$} & \footnotesize{$1$}\\
\hline
\footnotesize{$(p-1)p^{m-2}$} & \footnotesize{$p^m-p^{m-s}-1$}\\
\hline
\footnotesize{$(p^{m} - \epsilon  (p-1)\sqrt{p^*}^{m+s})(p-1)/p^2$} & \footnotesize{$p^{m-s-1} + \epsilon \eta_0^{m+1}(-1) (p-1) \sqrt{p^*}^{m-s-2} $}\\
\hline
\footnotesize{$(p^{m} +\epsilon  \sqrt{p^*}^{m+s})(p-1)/p^2$} & \footnotesize{$(p-1)(p^{m-s-1} -\epsilon \eta_0^{m+1}(-1) \sqrt{p^*}^{m-s-2})$}\\
\hline
\end{tabular}\end{center}	
\caption{
\label{Table0}  The weight distribution of  $\sC_{D_{0}}$ in Theorem \ref{TheoremD0} when $m+s$ is even}


\begin{center}	
\begin{tabular}{|c|c|c|}
\hline
Hamming weight $w$   & Multiplicity $A_w$ \\
\hline
\hline
\footnotesize{$0$} & \footnotesize{$1$}\\
\hline
\footnotesize{$ (p-1) (p^{m-2}-  \epsilon \sqrt{p}^{m+s-3})$} & \footnotesize{$p^{m-s-1} $}\\
\hline
\footnotesize{$ (p-1) p^{m-2}+  \epsilon2 \sqrt{p}^{m+s-3}  $} & \footnotesize{$  ( p^{m-s-1} + \epsilon    \sqrt{p}^{m-s-1})(p-1)/2 $}\\
\hline
\footnotesize{ $ (p-1)p^{m-2}$}& \footnotesize{$p^m-p^{m-s}-1+ ( p^{m-s-1} - \epsilon   \sqrt{p}^{m-s-1})(p-1)/2  $}\\
\hline
\end{tabular}\end{center}	
\caption{
\label{TableD1odd}  The weight distribution of  $\mathcal {C}_{D_{1}}$  when $p \equiv 1\pmod 4$ and $m+s$ is odd}

\begin{center}
\begin{tabular}{|c|c|c|}
\hline
Hamming weight $w$   & Multiplicity $A_w$ \\
\hline
\hline
\footnotesize{$0$} & \footnotesize{$1$}\\
\hline
\footnotesize{$ (p-1) (p^{m-2} + \epsilon \sqrt{p^*}^{m+s-3})$} & \footnotesize{$p^{m-s-1}  $}\\
\hline
\footnotesize{$ (p-1)p^{m-2} $} & \footnotesize{$ p^m-p^{m-s}-1+ ( p^{m-s-1} + \epsilon (-1)^m    \sqrt{p^*}^{m-s-1})(p-1)/2 $}\\
\hline
\footnotesize{ $ (p-1)p^{m-2}-  \epsilon2 \sqrt{p^*}^{m+s-3}$}& \footnotesize{$ ( p^{m-s-1} - \epsilon  (-1)^m    \sqrt{p^*}^{m-s-1})(p-1)/2   $}\\
\hline
\end{tabular}\end{center}	
\caption{
\label{TableD1oddd}  The weight distribution of  $\mathcal {C}_{D_{1}}$   when $p \equiv 3\pmod 4$ and $m+s$ is odd}

\begin{center}	
\begin{tabular}{|c|c|c|}
\hline
Hamming weight $w$   & Multiplicity $A_w$ \\
\hline
\hline
\footnotesize{$0$} & \footnotesize{$1$}\\
\hline
\footnotesize{$ (p-1) p^{m-2}$} & \footnotesize{$p^m-p^{m-s}-1$}\\
\hline
\footnotesize{$(p^{m} +\epsilon  \sqrt{p^*}^{m+s})(p-1)/p^2 $} 
& \footnotesize{$ p^{m-s-1} +(p^{m-s-1}+\epsilon \eta_0^{m+1}(-1)\sqrt{p^*}^{m-s-2})(p-1)/2 $}\\
\hline
\footnotesize{$((p-1)p^{m} -\epsilon (p+1)\sqrt{p^*}^{m+s})/p^2 $}
 & \footnotesize{$(p^{m-s-1} -\epsilon \eta_0^{m+1}(-1)\sqrt{p^*}^{m-s-2})(p-1)/2$}\\
\hline
\end{tabular}\end{center}	
\caption{
\label{TableD1even}  The weight distribution of   $\mathcal {C}_{D_{1}}$ when $m+s$ is even}
\end{table}
\begin{table}[!htp]
\begin{center}	
\begin{tabular}{|c|c|c|}
\hline
Hamming weight $\omega$   & Multiplicity $A_\omega$ \\
\hline
\hline
\footnotesize{$0$} & \footnotesize{$1$}\\
\hline
\footnotesize{$2(p-1)p^{m-2}$} & \footnotesize{$p^m-p^{m-s}-1$}\\
\hline
\footnotesize{$(p-1)(2p^{m-2}-\epsilon\eta_0(-1) \sqrt{p^*}^{m+s-3})$} & \footnotesize{$ p^{m-s-1}$}\\
\hline
\footnotesize{$2(p-1)p^{m-2} -  \epsilon (p-2-\eta_0(-1)) \sqrt{p^*}^{m+s-3} $} 
& \footnotesize{$(p^{m-s-1} +\epsilon  \eta_0^{m}(-1) \sqrt{p^*}^{m-s-1})(p-1)/2$}\\
\hline
\footnotesize{$2(p-1)p^{m-2} +  \epsilon (p-2+\eta_0(-1)) \sqrt{p^*}^{m+s-3}  $}
 & \footnotesize{$(p^{m-s-1} - \epsilon\eta_0^m(-1)   \sqrt{p^*}^{m-s-1})(p-1)/2$}\\
\hline
\end{tabular}\end{center}	
\caption{
\label{TableD01odd}  The weight distribution of  $\mathcal {C}_{D_{(0,1)}}$  when  $m+s$ is odd}
\begin{center}	
\begin{tabular}{|c|c|c|}
\hline
Hamming weight $\omega$   & Multiplicity $A_\omega$ \\
\hline
\hline
\footnotesize{$0$} & \footnotesize{$1$}\\
\hline
\footnotesize{$2(p-1)p^{m-2}$} & \footnotesize{$p^m-p^{m-s}-1$}\\
\hline
\footnotesize{$(2p^{m}-\epsilon(p-2)\sqrt{p^*}^{m+s})(p-1)/p^2$} & \footnotesize{$ p^{m-s-1} + \epsilon \eta_0^{m+1}(-1)(p-1) \sqrt{p^*}^{m-s-2}$}\\
\hline
\footnotesize{$2((p-1)p^{m}-\epsilon\sqrt{p^*}^{m+s})/p^2 $} 
& \footnotesize{$(p^{m-s-1} -\epsilon \eta_0^{m+1}(-1)\sqrt{p^*}^{m-s-2})(p-1)/2$}\\
\hline
\footnotesize{$2(p^{m}+ \epsilon \sqrt{p^*}^{m+s})(p-1)/p^2$}
 & \footnotesize{$(p^{m-s-1} -\epsilon \eta_0^{m+1}(-1)\sqrt{p^*}^{m-s-2})(p-1)/2$}\\
\hline
\end{tabular}\end{center}	
\caption{
\label{TableD01even}  The weight distribution of   $\mathcal {C}_{D_{(0,1)}}$ when $m+s$ is even}


\begin{center}	
\begin{tabular}{|c|c|c|}
\hline
Hamming weight $w$   & Multiplicity $A_w$ \\
\hline
\hline
\footnotesize{$0$} & \footnotesize{$1$}\\
\hline
\footnotesize{$(p-1)  ( p^{m-2} +  \epsilon  \sqrt{p}^{m+s-3})$} & \footnotesize{$p^{m-s-1}  $}\\
\hline
\footnotesize{$(p-1)p^{m-2} $} & \footnotesize{$p^m-p^{m-s}-1+(p^{m-s-1} +\epsilon  \sqrt{p}^{m-s-1})(p-1)/2$}\\
\hline
\footnotesize{ $ (p-1)p^{m-2} - \epsilon 2 \sqrt{p}^{m+s-3} $}& \footnotesize{$ ( p^{m-s-1} - \epsilon    \sqrt{p}^{m-s-1})(p-1)/2$}\\
\hline
\end{tabular}\end{center}	
\caption{
\label{TableD2}  The weight distribution of  $\mathcal {C}_{D_{2}}$   when $p \equiv 1\pmod 4$  and $m+s$ is odd}

\begin{center}	
\begin{tabular}{|c|c|c|}
\hline
Hamming weight $w$   & Multiplicity $A_w$ \\
\hline
\hline
\footnotesize{$0$} & \footnotesize{$1$}\\
\hline
\footnotesize{$(p-1) ( p^{m-2} -  \epsilon  \sqrt{p^*}^{m+s-3})$} & \footnotesize{$p^{m-s-1}  $}\\
\hline
\footnotesize{$ (p-1)p^{m-2} +  \epsilon 2 \sqrt{p^*}^{m+s-3} $} & \footnotesize{$ (p^{m-s-1} +\epsilon    (-1)^m    \sqrt{p^*}^{m-s-1})(p-1)/2 $}\\
\hline
\footnotesize{ $ (p-1)p^{m-2}$}& \footnotesize{$p^m-p^{m-s}-1+( p^{m-s-1} - \epsilon (-1)^m   \sqrt{p^*}^{m-s-1})(p-1)/2$}\\
\hline
\end{tabular}\end{center}	
\caption{
\label{TableD2odd}  The weight distribution of $\mathcal {C}_{D_{2}}$   when $p \equiv 3\pmod 4$  and $m+s$ is odd}

 \end{table}

\begin{table}
\begin{center}	
\begin{tabular}{|c|c|c|}
\hline
Hamming weight $\omega$   & Multiplicity $A_\omega$ \\
\hline
\hline
\footnotesize{$0$} & \footnotesize{$1$}\\
\hline
\footnotesize{$2(p-1)p^{m-2}$} & \footnotesize{$p^m-p^{m-s}-1$}\\
\hline
\footnotesize{$(p-1)(2p^{m-2} + \epsilon\eta_0(-1) \sqrt{p^*}^{m+s-3})$} & \footnotesize{$p^{m-s-1}$}\\
\hline
\footnotesize{$2(p-1)p^{m-2} -  \epsilon (p-2 + \eta_0(-1)) \sqrt{p^*}^{m+s-3} $} 
& \footnotesize{$(p^{m-s-1} +\epsilon  \eta_0^{m}(-1) \sqrt{p^*}^{m-s-1})(p-1)/2$}\\
\hline
\footnotesize{$2(p-1)p^{m-2} +  \epsilon (p-2 - \eta_0(-1)) \sqrt{p^*}^{m+s-3}  $}
 & \footnotesize{$(p^{m-s-1} - \epsilon\eta_0^m(-1)   \sqrt{p^*}^{m-s-1})(p-1)/2$}\\
\hline
\end{tabular}\end{center}	
\caption{
\label{Table02}  The weight distribution of   $\sC_{D_{(0,2)}}$  when  $m+s$ is odd}


\begin{center}	
\begin{tabular}{|c|c|c|}
\hline
Hamming weight $\omega$   & Multiplicity $A_\omega$ \\
\hline
\hline
\footnotesize{$0$} & \footnotesize{$1$}\\
\hline
\footnotesize{$2(p-1)p^{m-2}$} & \footnotesize{$p^{m-s-1}+p^m-p^{m-s}-1$}\\
\hline
\footnotesize{$2(p-1)p^{m-2} +\epsilon2  \sqrt{p^*}^{m+s-3}$} 
& \footnotesize{$(p^{m-s-1} +\epsilon  \eta_0^{m}(-1) \sqrt{p^*}^{m-s-1})(p-1)/2$}\\
\hline
\footnotesize{$2(p-1)p^{m-2} -\epsilon 2  \sqrt{p^*}^{m+s-3} $}
 & \footnotesize{$(p^{m-s-1} - \epsilon\eta_0^m(-1)   \sqrt{p^*}^{m-s-1})(p-1)/2$}\\
\hline
\end{tabular}\end{center}	
\caption{
\label{table12odd}  The weight distribution of  $\mathcal {C}_{D_{(1,2)}}$  when $m+s$ is odd}

\begin{center}	
\begin{tabular}{|c|c|c|}
\hline
Hamming weight $\omega$   & Multiplicity $A_\omega$ \\
\hline
\hline
\footnotesize{$0$} & \footnotesize{$1$}\\
\hline
\footnotesize{$2(p-1)p^{m-2}$} & \footnotesize{$p^m-p^{m-s}-1$}\\
\hline
\footnotesize{$2(p-1)(p^{m}+\epsilon\sqrt{p^*}^{m+s})/p^2$} & \footnotesize{$ p^{m-s-1} + \epsilon \eta_0^{m+1}(-1)(p-1) \sqrt{p^*}^{m-s-2}$}\\
 
\hline
\footnotesize{$2((p-1)p^{m}- \epsilon\sqrt{p^*}^{m+s})/p^2$}
 & \footnotesize{$(p-1)(p^{m-s-1} -\epsilon \eta_0^{m+1}(-1)\sqrt{p^*}^{m-s-2})$}\\
\hline
\end{tabular}\end{center}	
\caption{
\label{table12}  The weight distribution of   $\mathcal {C}_{D_{(1,2)}}$ when $m+s$ is even}

\begin{center}	
\begin{tabular}{|c|c|c|}
\hline
Hamming weight $\omega$   & Multiplicity $A_\omega$ \\
\hline
\hline
\footnotesize{$0$} & \footnotesize{$1$}\\
\hline
\footnotesize{$ (p^{m-2}-  \epsilon \sqrt{p}^{m+s-3})(p-1)^2/2$} & \footnotesize{$p^{m-s-1} $}\\
\hline
\footnotesize{$(p-1)(p^{m-2} (p-1)/2+  \epsilon  \sqrt{p}^{m+s-3})  $} & \footnotesize{$  ( p^{m-s-1} + \epsilon    \sqrt{p}^{m-s-1})(p-1)/2 $}\\
\hline
\footnotesize{ $p^{m-2} (p-1)^2/2  $}& \footnotesize{$p^m-p^{m-s}-1+ ( p^{m-s-1} - \epsilon   \sqrt{p}^{m-s-1})(p-1)/2  $}\\
\hline
\end{tabular}\end{center}	
\caption{
\label{TableSQodd}  The weight distribution of  $\mathcal {C}_{D_{sq}}$   when $p \equiv 1\pmod 4$ and $m+s$ is odd}
 \end{table}

\begin{table}
\begin{center}
\begin{tabular}{|c|c|c|}
\hline
Hamming weight $\omega$   & Multiplicity $A_\omega$ \\
\hline
\hline
\footnotesize{$0$} & \footnotesize{$1$}\\
\hline
\footnotesize{$ (p^{m-2} + \epsilon \sqrt{p^*}^{m+s-3})(p-1)^2/2$} & \footnotesize{$p^{m-s-1}  $}\\
\hline
\footnotesize{$p^{m-2} (p-1)^2/2 $} & \footnotesize{$ p^m-p^{m-s}-1+ ( p^{m-s-1} + \epsilon (-1)^m    \sqrt{p^*}^{m-s-1})(p-1)/2 $}\\
\hline
\footnotesize{ $ (p-1)(p^{m-2} (p-1)/2 -  \epsilon \sqrt{p^*}^{m+s-3})$}& \footnotesize{$ ( p^{m-s-1} - \epsilon  (-1)^m    \sqrt{p^*}^{m-s-1})(p-1)/2   $}\\
\hline
\end{tabular}\end{center}	
\caption{
\label{TableSQoddd}  The weight distribution of  $\mathcal {C}_{D_{sq}}$   when $p \equiv 3\pmod 4$ and $m+s$ is odd}

\begin{center}	
\begin{tabular}{|c|c|c|}
\hline
Hamming weight $\omega$   & Multiplicity $A_\omega$ \\
\hline
\hline
\footnotesize{$0$} & \footnotesize{$1$}\\
\hline
\footnotesize{$p^{m} (p-1)^2/2p^2$} & \footnotesize{$p^m-p^{m-s}-1$}\\
\hline
\footnotesize{$(p^{m} +\epsilon  \sqrt{p^*}^{m+s})(p-1)^2/2p^2$} 
& \footnotesize{$ p^{m-s-1} +(p^{m-s-1}+\epsilon \eta_0^{m+1}(-1)\sqrt{p^*}^{m-s-2})(p-1)/2 $}\\
\hline
\footnotesize{$((p-1)p^{m} -\epsilon (p+1)\sqrt{p^*}^{m+s})(p-1)/2p^2$}
 & \footnotesize{$(p^{m-s-1} -\epsilon \eta_0^{m+1}(-1)\sqrt{p^*}^{m-s-2})(p-1)/2$}\\
\hline
\end{tabular}\end{center}	
\caption{
\label{TableSQeven}  The weight distribution of   $\mathcal {C}_{D_{sq}}$ when $m+s$ is even}


\begin{center}	
\begin{tabular}{|c|c|c|}
\hline
Hamming weight $w$   & Multiplicity $A_w$ \\
\hline
\hline
\footnotesize{$0$} & \footnotesize{$1$}\\
\hline
\footnotesize{$p^{m-2}(p^2-1)/2-\epsilon  \sqrt{p}^{m+s-3} (p-1)^2/2$} & \footnotesize{$p^{m-s-1}  $}\\
\hline
\footnotesize{$p^{m-2}(p^2-1)/2 $} & \footnotesize{$p^m-p^{m-s}-1+ (p^{m-s-1} +\epsilon    \sqrt{p}^{m-s-1})(p-1)/2$}\\
\hline
\footnotesize{ $p^{m-2}(p^2-1)/2+ \epsilon (p-1) \sqrt{p}^{m+s-3} $}& \footnotesize{$( p^{m-s-1} - \epsilon \sqrt{p}^{m-s-1})(p-1)/2 $}\\
\hline
\end{tabular}\end{center}	
\caption{
\label{tableSQ0}  The weight distribution of  $\mathcal {C}_{D_{(sq,0)}}$   when $p \equiv 1\pmod 4$ and $m+s$ is odd}
\begin{center}	
\begin{tabular}{|c|c|c|}
\hline
Hamming weight $w$   & Multiplicity $A_w$ \\
\hline
\hline
\footnotesize{$0$} & \footnotesize{$1$}\\
\hline
\footnotesize{$p^{m-2}(p^2-1)/2+\epsilon  \sqrt{p^*}^{m+s-3}(p-1)^2/2$} & \footnotesize{$p^{m-s-1}   $}\\
\hline
\footnotesize{$p^{m-2}(p^2-1)/2- \epsilon(p-1)\sqrt{p^*}^{m+s-3} $} & \footnotesize{$(p^{m-s-1} +\epsilon   (-1)^m   \sqrt{p^*}^{m-s-1})(p-1)/2$}\\
\hline
\footnotesize{ $p^{m-2}(p^2-1)/2$}& \footnotesize{$p^m-p^{m-s}-1+( p^{m-s-1} - \epsilon (-1)^m    \sqrt{p^*}^{m-s-1})(p-1)/2$}\\
\hline
\end{tabular}\end{center}	
\caption{
\label{tableSQ00}  The weight distribution of  $\mathcal {C}_{D_{(sq,0)}}$   when $p \equiv 3\pmod 4$ and $m+s$ is odd}
\end{table}

\begin{table}
\begin{center}	
\begin{tabular}{|c|c|c|}
\hline
Hamming weight $w$   & Multiplicity $A_w$ \\
\hline
\hline
\footnotesize{$0$} & \footnotesize{$1$}\\
\hline
\footnotesize{$p^{m}(p^2-1)/2p^2  $} & \footnotesize{$p^m-p^{m-s}-1$}\\
\hline
\footnotesize{$((p+1)p^{m} -\epsilon(p-1)\sqrt{p^*}^{m+s})(p-1)/2p^2$} 
& \footnotesize{$ p^{m-s-1} +\frac{1}{2}(p-1)(p^{m-s-1}+\epsilon \eta_0^{m+1}(-1)\sqrt{p^*}^{m-s-2}) $}\\
\hline
\footnotesize{$ (p^{m} +\epsilon  \sqrt{p^*}^{m+s})(p^2-1)/2p^2$}
 & \footnotesize{$\frac{1}{2}(p-1)(p^{m-s-1} -\epsilon \eta_0^{m+1}(-1)\sqrt{p^*}^{m-s-2})$}\\
\hline
\end{tabular}\end{center}	
\caption{
\label{tableSQ0even}  The weight distribution of   $\mathcal {C}_{D_{(sq,0)}}$ when $m+s$ is even}


\begin{center}	
\begin{tabular}{|c|c|c|}
\hline
Hamming weight $\omega$   & Multiplicity $A_\omega$ \\
\hline
\hline
\footnotesize{$0$} & \footnotesize{$1$}\\
\hline
\footnotesize{$ ( p^{m-2} +  \epsilon  \sqrt{p}^{m+s-3})(p-1)^2/2 $} & \footnotesize{$p^{m-s-1}  $}\\
\hline
\footnotesize{$p^{m-2} (p-1)^2/2$} & \footnotesize{$p^m-p^{m-s}-1+(p^{m-s-1} +\epsilon  \sqrt{p}^{m-s-1})(p-1)/2$}\\
\hline
\footnotesize{ $(p-1) (p^{m-2} (p-1)/2 - \epsilon \sqrt{p}^{m+s-3}) $}& \footnotesize{$ ( p^{m-s-1} - \epsilon    \sqrt{p}^{m-s-1})(p-1)/2$}\\
\hline
\end{tabular}\end{center}	
\caption{
\label{tableNSQ}  The weight distribution of  $\mathcal {C}_{D_{nsq}}$   when $p \equiv 1\pmod 4$  and $m+s$ is odd}

\begin{center}	
\begin{tabular}{|c|c|c|}
\hline
Hamming weight $\omega$   & Multiplicity $A_\omega$ \\
\hline
\hline
\footnotesize{$0$} & \footnotesize{$1$}\\
\hline
\footnotesize{$ ( p^{m-2} -  \epsilon  \sqrt{p^*}^{m+s-3})(p-1)^2/2 $} & \footnotesize{$p^{m-s-1}  $}\\
\hline
\footnotesize{$(p-1)(p^{m-2} (p-1)/2 +  \epsilon  \sqrt{p^*}^{m+s-3}) $} & \footnotesize{$ (p^{m-s-1} +\epsilon    (-1)^m    \sqrt{p^*}^{m-s-1})(p-1)/2 $}\\
\hline
\footnotesize{ $p^{m-2} (p-1)^2/2  $}& \footnotesize{$p^m-p^{m-s}-1+( p^{m-s-1} - \epsilon (-1)^m   \sqrt{p^*}^{m-s-1})(p-1)/2$}\\
\hline
\end{tabular}\end{center}	
\caption{
\label{tableNSQ3}  The weight distribution of  $\mathcal {C}_{D_{nsq}}$   when $p \equiv 3\pmod 4$  and $m+s$ is odd}


\begin{center}	
\begin{tabular}{|c|c|c|}
\hline
Hamming weight $w$   & Multiplicity $A_w$ \\
\hline
\hline
\footnotesize{$0$} & \footnotesize{$1$}\\
\hline
\footnotesize{$p^{m-2}(p^2-1)/2 + \epsilon  \sqrt{p}^{m+s-3} (p-1)^2/2$} & \footnotesize{$p^{m-s-1}  $}\\
\hline
\footnotesize{$p^{m-2}(p^2-1)/2 - \epsilon (p-1) \sqrt{p}^{m+s-3} $} & \footnotesize{$(p^{m-s-1} +\epsilon    \sqrt{p}^{m-s-1})(p-1)/2$}\\
\hline
\footnotesize{ $p^{m-2}(p^2-1)/2$}& \footnotesize{$p^m-p^{m-s}-1+ ( p^{m-s-1} - \epsilon \sqrt{p}^{m-s-1})(p-1)/2 $}\\
\hline
\end{tabular}\end{center}	
\caption{
\label{tableNSQ0}  The weight distribution of  $\mathcal {C}_{D_{(nsq,0)}}$   when $p \equiv 1\pmod 4$ and $m+s$ is odd}
\end{table}

\begin{table}

\begin{center}	
\begin{tabular}{|c|c|c|}
\hline
Hamming weight $w$   & Multiplicity $A_w$ \\
\hline
\hline
\footnotesize{$0$} & \footnotesize{$1$}\\
\hline
\footnotesize{$p^{m-2}(p^2-1)/2 - \epsilon \sqrt{p^*}^{m+s-3}(p-1)^2/2$} & \footnotesize{$p^{m-s-1}   $}\\
\hline
\footnotesize{$p^{m-2}(p^2-1)/2 $} & \footnotesize{$p^m-p^{m-s}-1+(p^{m-s-1} +\epsilon   (-1)^m   \sqrt{p^*}^{m-s-1})(p-1)/2$}\\
\hline
\footnotesize{ $p^{m-2}(p^2-1)/2 + \epsilon(p-1)\sqrt{p^*}^{m+s-3}$}& \footnotesize{$( p^{m-s-1} - \epsilon (-1)^m    \sqrt{p^*}^{m-s-1})(p-1)/2$}\\
\hline
\end{tabular}\end{center}	
\caption{
\label{tableNSQ03}  The weight distribution of  $\mathcal {C}_{D_{(nsq,0)}}$   when $p \equiv 3\pmod 4$ and $m+s$ is odd}


\begin{center}	
\begin{tabular}{|c|c|c|}
\hline
Hamming weight $\omega$   & Multiplicity $A_\omega$ \\
\hline
\hline
\footnotesize{$0$} & \footnotesize{$1$}\\
\hline
\footnotesize{$p^{m-2}$} & \footnotesize{$p^m-p^{m-s}-1$}\\
\hline
\footnotesize{$p^{m-2} - \epsilon  (p-1)\sqrt{p^*}^{m+s-4}$} & \footnotesize{$p^{m-s-1} + \epsilon \eta_0^{m+1}(-1) (p-1) \sqrt{p^*}^{m-s-2} $}\\
\hline
\footnotesize{$ p^{m-2} +\epsilon  \sqrt{p^*}^{m+s-4}$} & \footnotesize{$(p-1)(p^{m-s-1} -\epsilon \eta_0^{m+1}(-1) \sqrt{p^*}^{m-s-2})$}\\
\hline
\end{tabular}\end{center}	
\caption{
\label{tablePunc}  The weight distribution of  $\sC_{\overline{D}_{0}}$ when $m+s$ is even}


\begin{center}	
\begin{tabular}{|c|c|c|}
\hline
Hamming weight $\omega$   & Multiplicity $A_\omega$ \\
\hline
\hline
\footnotesize{$0$} & \footnotesize{$1$}\\
\hline
\footnotesize{$2\cdot 3^{m-2}$} & \footnotesize{$3^{m-s-1}+3^m-3^{m-s}-1$}\\
\hline
\footnotesize{$2\cdot 3^{m-2} +\epsilon  \sqrt{-3}^{m+s-3}$} 
& \footnotesize{$3^{m-s-1} +\epsilon  (-1)^{m} \sqrt{-3}^{m-s-1}$}\\
\hline
\footnotesize{$2\cdot 3^{m-2} -\epsilon   \sqrt{-3}^{m+s-3} $}
 & \footnotesize{$3^{m-s-1} - \epsilon(-1)^{m}  \sqrt{-3}^{m-s-1}$}\\
\hline
\end{tabular}\end{center}	
\caption{
\label{table12oddpunc}  The weight distribution of  $\mathcal {C}_{\overline{D}_{(1,2)}}$  when $p=3$ and  $m+s$ is odd}

\begin{center}	
\begin{tabular}{|c|c|c|}
\hline
Hamming weight $\omega$   & Multiplicity $A_\omega$ \\
\hline
\hline
\footnotesize{$0$} & \footnotesize{$1$}\\
\hline
\footnotesize{$2\cdot 3^{m-2}$} & \footnotesize{$3^m-3^{m-s}-1$}\\
\hline
\footnotesize{$2(3^{m-2}+\epsilon\sqrt{-3}^{m+s-4})$} & \footnotesize{$ 3^{m-s-1} + \epsilon (-1)^{m+1}2 \sqrt{-3}^{m-s-2}$}\\
 
\hline
\footnotesize{$2\cdot 3^{m-2}- \epsilon\sqrt{-3}^{m+s-4}$}
 & \footnotesize{$2(3^{m-s-1} -\epsilon(-1)^{m+1}\sqrt{-3}^{m-s-2})$}\\
\hline
\end{tabular}\end{center}	
\caption{
\label{table12punc}  The weight distribution of   $\mathcal {C}_{\overline{D}_{(1,2)}}$ when $p=3$ and $m+s$ is even}
\end{table}

\end{document}